\documentclass[12pt]{article}
\usepackage{graphicx}
\usepackage{color}
\usepackage{amsmath, amssymb, amsthm, ascmac, bm, booktabs, caption, comment, enumerate, latexsym, lscape, longtable, mathabx, mathrsfs, mathtools, multirow, natbib, newtxtext, psfrag, setspace, subcaption, thmtools}
\usepackage[stable]{footmisc}
\definecolor{mycolor}{cmyk}{0.85, 0.21, 0, 0.06}
\usepackage[colorlinks=true, linkcolor=mycolor, filecolor=mycolor, urlcolor=mycolor, citecolor=mycolor, backref = page]{hyperref}
\usepackage[margin = 1in]{geometry}

\DeclareMathOperator*{\bE}{\mathbb{E}}
\DeclareMathOperator*{\bP}{\mathbb{P}}

\renewcommand{\hat}{\widehat}
\renewcommand{\tilde}{\widetilde}

\allowdisplaybreaks[1]

\theoremstyle{definition}
\newtheorem{theorem}{Theorem}[section]
\newtheorem{assumption}{Assumption}[section]
\newtheorem{definition}{Definition}[section]

\newtheorem{lemma}{Lemma}[section]

\newtheorem{algorithm}{Procedure}[section]

\declaretheorem[style=definition]{example}
\declaretheorem[style=definition]{remark}

\numberwithin{equation}{section}

\title{An Effective Treatment Approach to Difference-in-Differences \\
with General Treatment Patterns\thanks{
	The author is very grateful to Tadao Hoshino, Takuya Ishihara, and Ryo Okui for their helpful comments, which significantly improved the paper.
	The author thanks Lei Qin for his excellent research assistance.
	This work was supported by JSPS KAKENHI Grant Number 20K13469.
	The author reports no other competing interests to declare.
	The paper was previously circulated under the title ``Doubly Robust Difference-in-Differences with General Treatment Patterns''.
	}
}

\author{Takahide Yanagi\thanks{
		Graduate School of Economics, Kyoto University, Yoshida Honmachi, Sakyo, Kyoto 606-8501, Japan. Email: \href{mailto:yanagi@econ.kyoto-u.ac.jp}{yanagi@econ.kyoto-u.ac.jp}
		}
}

\date{First version: December 2022 \hspace{0.5in} This version: May 2023}

\begin{document}

\renewcommand\thmcontinues[1]{Continued}
	
\doublespacing

\maketitle

\begin{abstract}
	We consider a general difference-in-differences model in which the treatment variable of interest may be non-binary and its value may change in each period.
	It is generally difficult to estimate treatment parameters defined with the potential outcome given the entire path of treatment adoption, because each treatment path may be experienced by only a small number of observations.
	We propose an alternative approach using the concept of effective treatment, which summarizes the treatment path into an empirically tractable low-dimensional variable, and develop doubly robust identification, estimation, and inference methods.
	We also provide a companion R software package.
	
	\medskip
	
	
	\medskip 
	
\end{abstract}
	
\clearpage

\pagestyle{plain}
	
\section{Introduction} \label{sec:introduction}

Difference-in-differences (DiD) is one of the most popular empirical strategies using panel data.
Under the parallel trends assumption, one can identify meaningful treatment parameters by comparing the time evolution of the outcomes between the treatment and comparison groups.
Most of the recent literature has been devoted to the development of methods in the staggered adoption setting, in which each unit continues to receive a binary treatment after the first treatment receipt (cf. \citealp{de2022two}; \citealp{sun2022linear}; \citealp{callaway2023difference}; \citealp{roth2023s}).

Although staggered adoption is a typical setting in applied research, there are many empirical situations in which this is not the case.\footnote{
	\cite{de2022difference} conduct a survey of the 100 most-cited articles in the American Economic Review from 2015 to 2019.
	They document that 26 papers estimate a two-way fixed effects (TWFE) regression, but only 4 papers fit the standard staggered adoption setting.
}
For example, \cite{vella1998whose} estimate the effect of union membership on wages.
Approximately 28\% of the individuals in their sample change union membership status at least twice, implying non-staggered treatment adoption.
In another example, \cite{deryugina2017fiscal} examines the impact of being hit by a hurricane on the fiscal cost for a US county.
In this scenario, the treatment variable of interest indicates whether (or the number of times) a county was hit by a hurricane in a year.
This empirical situation also does not fit perfectly into the standard setting of staggered adoption because a county damaged by a hurricane this year is not necessarily hit by a hurricane in the following year.

We consider a general DiD model in which the treatment variable of interest may be non-binary and time-varying, in that the treatment realization of each unit may change in each period, which encompasses the staggered setting as a special case.
It is generally difficult to estimate sensible treatment parameters using DiD methods conditional on the entire path of treatment adoption.
This is because there may be a large number of treatment paths, so that each treatment path may be experienced by only a small set of observations.
As a concrete example, in the dataset of \cite{vella1998whose}, there are 95 different paths of union membership status, 53 of which are experienced by only one individual.
In such a case, it is challenging to develop a DiD method that can accurately estimate meaningful treatment parameters.

We propose a DiD method using the concept of \textit{effective treatment}, which converts the entire treatment path into an empirically tractable low-dimensional variable.
In an influential study, \cite{manski2013identification} develops the concept of effective treatment to summarize potentially complicated treatment spillovers among individuals in the presence of social interactions.\footnote{
	The effective treatment is also referred to as \textit{exposure mapping} in the literature on causal inference under interference (cf., \citealp{aronow2021spillover}).
}
Building on this idea, we propose an effective treatment approach to summarize the treatment path over time in the DiD setting, while assuming no interference between units.
In particular, we consider the following effective treatment specifications useful for estimating the instantaneous and dynamic treatment effects: whether some treatment realization has been received at least once so far, the earliest period of receiving some treatment realization so far, and the number of treatment adoptions so far.

Given an effective treatment specification, we set our target parameter as the average treatment effect for \textit{movers} (ATEM), who switch from the comparison group to receiving an effective treatment realization between two periods.
We provide a set of sufficient conditions, including a conditional parallel trends assumption given observed covariates, under which the ATEM is identified by each of outcome regression (OR), inverse probability weighting (IPW), and doubly robust (DR) estimands.
An important feature of our identification analysis is that it allows the chosen effective treatment specification to be ``incorrect'' in that the potential outcome is not well defined with respect to this effective treatment.
This feature is empirically attractive, because researchers generally do not have prior knowledge of the true form of effective treatment. 
Furthermore, as the parallel trends assumption is key to our analysis, we also derive its testable implication, which suggests a ``pre-trends'' test similar to existing ones.

Based on the identification result using the DR estimand, we propose a two-step estimation procedure.
First, we estimate the OR function for stayers and a generalized propensity score (GPS) under parametric assumptions.
Then, we construct an ATEM estimator using the sample counterpart of the identification result.
Desirably, the resulting ATEM estimator has the DR property: it is consistent and asymptotically normally distributed if either the parametric assumptions of the OR function or GPS are correct.
We also consider a multiplier bootstrap inference procedure for constructing uniform confidence bands (UCB).

As an empirical illustration, we apply our methods to the dataset of \cite{vella1998whose}.
Using each effective treatment specification discussed above, we find no statistical evidence to support significant effects of union membership on current and future wages.
Furthermore, the pre-trends testing result is consistent with the parallel trends assumption.

\bigskip 

This study builds on a number of recent contributions, including \cite{de2020two}, \cite{goodman2021difference}, \cite{imai2021use}, \cite{sun2021estimating}, \cite{wooldridge2021two}, \cite{athey2022design}, and \cite{borusyak2023revisiting}. 
Among them, our methods extend in particular \cite{sant2020doubly} and \cite{callaway2021difference}, who develop the DiD methods using the DR estimands in the canonical two-period model and the staggered setting, to our situation with possibly non-binary time-varying treatment.

Another closely related study is \cite{de2022difference}, who propose DiD estimation for certain instantaneous and dynamic treatment effects in a setting similar to ours.
They begin by showing that commonly used TWFE regressions produce biased estimates for treatment effects.
Importantly, this negative result is applicable to our situation, and thus, we cannot recover sensible treatment parameters from such TWFE regressions.
See \cite{de2020two,de2022several}, \cite{goodman2021difference}, and \cite{ishimaru2022what} for related results.
\cite{de2022difference} then propose an event study estimation method by defining the ``event'' as the period in which the unit receives a treatment realization for the first time.
In comparison, we consider general forms of effective treatment, and our methods have an empirically attractive DR property in the same spirits as \cite{sant2020doubly} and \cite{callaway2021difference}.

This study is also related to several recent works that propose DiD methods when the treatment variable of interest is non-binary.
\cite{callaway2021continuous} and \cite{de2022continuous} consider settings in which the treatment variable has continuous realizations.
\cite{de2022several} study DiD estimation in the presence of several binary treatment variables with staggered adoption.
These studies are tailored to specific forms of treatment realizations, whereas our effective treatment approach handles essentially arbitrary treatment realizations.

\bigskip 

The rest of this paper is organized as follows.
In Section \ref{sec:setup}, we introduce the setup and formalize the concept of effective treatment.
Sections \ref{sec:identification} and \ref{sec:estimation} develop the identification and estimation methods.
Section \ref{sec:testable} presents the testable implication for parallel trends.
Section \ref{sec:application} provides the empirical illustration.
The supplementary appendix contains technical proofs and simulation results.
A companion R package is available from the author's website.

\section{Setup and Effective Treatment} \label{sec:setup}

We have panel data $\{ (Y_{it}, D_{it}, X_i): i \in \mathcal{N}, t \in \mathcal{T} \}$, where $Y_{it} \in \mathcal{Y} \subseteq \mathbb{R}$ and $D_{it} \in \mathcal{D} \subseteq \mathbb{R}^{\dim(D)}$ are the outcome and treatment, respectively, for unit $i$ in period $t$, $X_i \in \mathcal{X} \subseteq \mathbb{R}^{\dim(X)}$ is a vector of unit-$i$-specific covariates, and $\mathcal{N} = \{ 1, \dots, N \}$ and $\mathcal{T} = \{ 1, \dots, T \}$ denote the sets of units and periods, respectively.
Some unobserved factors may be correlated with $Y_{it}$ and $D_{it}$, implying the possible treatment endogeneity.

The treatment $D_{it}$ may be binary, discrete, continuous, or multidimensional.
When $D_{it} = 0$, unit $i$ receives no treatment in period $t$ (throughout the paper, for notational simplicity, we write a generic vector of zeros by $0$).
By contrast, $D_{it} = d$ with $d \neq 0$ indicates that unit $i$ receives treatment intensity $d$ in period $t$.
The realization of $D_{it}$ can vary over time in that each unit can move into and out of receiving each treatment intensity in each period.
For each $t$, let $\bm{D}_{it} = (D_{i1}, \dots, D_{it}) \in \mathcal{D}^t$.

Let $Y_{it}(\bm{d}_T)$ denote the potential outcome given the entire treatment path $\bm{D}_{iT} = \bm{d}_T$.
This notation makes explicit that the current outcome may depend on past, current, and future treatment intensity.
For example, $Y_{it}(0)$ denotes the potential outcome in period $t$ if unit $i$ has never been treated in $T$ periods.
By construction, $Y_{it} = Y_{it}(\bm{D}_{iT})$.

\begin{remark}[No anticipation] \label{remark:noanticipate}
	The no-anticipation condition is commonly assumed in the literature, under which the current outcome does not depend on future treatment adoption.
	Then, the potential outcome in period $t$ given $\bm{D}_{it} = \bm{d}_t$ is well defined.
	We will reflect our belief in no anticipation when choosing the functional form of effective treatment.
\end{remark}

In general, it is difficult to estimate treatment parameters defined with respect to $Y_{it}(\bm{d}_T)$.
This is because the number of units that follow a specific treatment path $\bm{d}_T$ may be small in practice, especially when there is much variation in treatment intensity or when the length of the time series is not very small.
For example, when $D_{it}$ has $\bar d$ discrete realizations, there may be $(\bar d)^T$ potential outcomes for each $i$ and $t$.
In such a situation, only a limited number of units would experience $\bm{d}_T$, making it difficult to obtain precise causal estimates.

To circumvent this problem, we adopt the concept of effective treatment, which converts the treatment path $\bm{D}_{iT}$ into an empirically tractable low-dimensional variable $E_{it}$ that researchers believe to be relevant to the outcome in period $t$.
Specifically, for each $t$, consider a user-specified function $E_t: \mathcal{D}^T \to \mathcal{E}_t$, where $\mathcal{E}_t \subset \mathbb{R}^{\dim(E)}$ denotes the range of $E_t$.
Here, the functional form of $E_t$ may change depending on $t$.
We call $E_t$ the \textit{effective treatment function}, the corresponding realization $E_{it} \coloneqq E_t(\bm{D}_{iT})$ the \textit{realized effective treatment} for unit $i$ in period $t$, and the values in $\mathcal{E}_t$ the \textit{effective treatment intensity}.
When $E_{it} = 0$, we consider unit $i$ in period $t$ as the comparison unit under this $E_t$.
We write the set of nonzero realizations of $E_t$ as $\tilde{\mathcal{E}}_t \coloneqq \mathcal{E}_t \setminus \{ 0 \}$.

The following is a set of examples of the functional form of effective treatment.

\begin{example}[name = The \textit{once} specification, label = ex:once]
	If we are interested in the relationship between the outcome and whether the unit has received some treatment intensity at least \textit{once} so far, we may consider
	\begin{align} \label{eq:once}
		E_{it}^{\mathrm{O}} = \bm{1}\{ \bm{D}_{it} \neq 0 \}.
	\end{align}
	The set of nonzero realizations is $\tilde{\mathcal{E}}_t^{\mathrm{O}} = \{ 1 \}$.
\end{example}

\begin{example}[name = The \textit{event} specification, label = ex:event]
	When the period so far in which each unit is treated for the first time is relevant to the outcome, we may use the following specification:
	\begin{align} \label{eq:event}
		E_{it}^{\mathrm{E}} = 
		\begin{cases}
			\min \{ s \le t : D_{is} \neq 0 \} & \text{if $\bm{D}_{it} \neq 0$} \\
			0 &  \text{otherwise}
		\end{cases}
	\end{align}
	We call $E_{it}^{\mathrm{E}}$ and $t - E_{it}^{\mathrm{E}}$ the \textit{event date} and \textit{event time}, respectively (cf. \citealp{miller2023introductory}).
	The set of nonzero realizations in period $t$ is $\tilde{\mathcal{E}}_t^{\mathrm{E}} = \{ 1, \dots, t \}$.
\end{example}

\begin{example}[name = The \textit{number} specification, label = ex:number]
	When the \textit{number} of treatment adoptions so far matters for the outcome, we set
	\begin{align} \label{eq:number}
		E_{it}^{\mathrm{N}} = \sum_{s=1}^t \bm{1}\{ D_{is} \neq 0 \},
	\end{align}
	with the corresponding $\tilde{\mathcal{E}}_t^{\mathrm{N}} = \{ 1, \dots, t \}$.
\end{example}

\begin{remark}[Limited anticipation] \label{remark:limited}
	Under limited anticipation, there is some known $\delta > 0$ such that the outcome in period $t$ is determined from the path of treatment adoptions over $t + \delta$ periods (cf. \citealp{callaway2021difference}).
	We can easily incorporate our belief in it into the functional form of effective treatment.
	For example, \eqref{eq:once} can be modified to $E_{it}^{\mathrm{O}} = \bm{1}\{ \bm{D}_{i,t+\delta} \neq 0 \}$.
\end{remark}

Next, we formalize the correctness of the effective treatment specification.
The following definition builds on the treatment effects literature on cross-unit interference (cf., \citealp{manski2013identification}; \citealp{vazquez2022identification}; \citealp{hoshino2023randomization}).

\begin{definition}
	Consider effective treatment functions $E_t^*: \mathcal{D}^T \to \mathcal{E}_t^*$ and $E_t: \mathcal{D}^T \to \mathcal{E}_t$.
	\begin{enumerate}[(i)]
		\item $E_t^*$ is \textit{correct} if for all $i \in \mathcal{N}$ and $\bm{d}_T, \bm{d}_T' \in \mathcal{D}^{T}$, $E_t^*(\bm{d}_T) = E_t^*(\bm{d}_T')$ implies $Y_{it}(\bm{d}_T) = Y_{it}(\bm{d}_T')$.
		\item $E_t^*$ is \textit{coarser} than $E_t$ if there is a surjective mapping $c_t: \mathcal{E}_t \to \mathcal{E}_t^*$ such that $E_t^*(\bm{d}_T) = c_t(E_t(\bm{d}_T))$ for any $\bm{d}_T \in \mathcal{D}^T$.
		\item A correct $E_t^*$ is \textit{true} if $E_t^*$ is coarser than any other correct specification.
	\end{enumerate}
\end{definition}

There may be several correct specifications.
For example, if one of \eqref{eq:once}--\eqref{eq:number} is correct, so is the identity mapping $E_t(\bm{D}_{iT}) = \bm{D}_{iT}$, which is the ``finest'' correct specification.
As the identity mapping is always correct, the true specification always exists, which may or may not equal the identity mapping.

\begin{example}
	Since $E_{it}^{\mathrm{O}} = \bm{1}\{ E_{it}^{\mathrm{E}} \neq 0 \} = \bm{1}\{ E_{it}^{\mathrm{N}} \neq 0 \}$, the once specification is coarser than the event and number specifications.
	In general, there is no relationship between the coarseness of the event and number specifications.
\end{example}

Throughout the paper, we write the true effective treatment function as $E_t^*: \mathcal{D}^T \to \mathcal{E}_t^*$, which is generally unknown to researchers.
Let $E_{it}^* = E_t^*(\bm{D}_{iT})$ denote the realized true effective treatment for unit $i$ in period $t$.
From the definition of the correct specification, we can define the potential outcome for unit $i$ in period $t$ given $E_{it}^* = e$, that is, $Y_{it}^*(e) = Y_{it}(\bm{d}_T)$ for any $\bm{d}_T \in \mathcal{D}^T$ such that $E_t^*(\bm{d}_T) = e$.
For example, $Y_{it}^*(0)$ denotes the potential outcome when unit $i$ in period $t$ is a comparison unit under the true $E_t^*$.
The individual treatment effect in period $t$ can be written as $Y_{it}^*(e) - Y_{it}^*(0)$, where $e \in \tilde{\mathcal{E}}_t^*$ and $\tilde{\mathcal{E}}_t^* \coloneqq \mathcal{E}_t^* \setminus \{ 0 \}$, which may be heterogeneous across $i$, $t$, and $e$.

Ideally, we would like to adopt the true effective treatment function $E_t^*$ for the DiD analysis to obtain meaningful causal estimates.
In reality, however, we generally have no prior knowledge of the true specification, unless the canonical two-period model or the staggered case.
Furthermore, even if we know a correct specification, such as identity mapping, it may be too ``fine'' to obtain accurate DiD estimates.
Thus, instead of pursing the true specification, we explicitly allow for potential misspecification and aim to obtain relatively precise DiD estimates equipped with sensible causal interpretation.
Similar approaches have been advocated in the literature of causal inference in the presence of treatment spillovers (\citealp{aronow2017estimating}; \citealp{leung2022causal}; \citealp{vazquez2022identification}; \citealp{hoshino2023causal}).

\section{Identification Analysis} \label{sec:identification}

We begin by imposing the following condition:

\begin{assumption}[Effective treatment 1] \label{as:effective}
	Let $E_t: \mathcal{D}^T \to \mathcal{E}_t$ be a user-specified effective treatment function.
	\begin{enumerate}[(i)]
		\item For each $t \in \mathcal{T}$, $\mathcal{E}_t$ is a finite set on $\mathbb{R}^{\dim(E)}$. 
		\item For any $t \in \mathcal{T}$ and $\bm{d}_T \in \mathcal{D}^T$, $E_t(\bm{d}_T) = 0$ implies $E_t^*(\bm{d}_T) = 0$.
	\end{enumerate}
\end{assumption}

Assumption \ref{as:effective}(i) requires the pre-specified $E_t$ to have discrete realizations.
Under Assumption \ref{as:effective}(ii), any comparison unit under the pre-specified $E_t$ is a comparison unit even under the true $E_t^*$.
This should be a mild requirement, because each unit $i$ should be considered as a comparison unit when $\bm{D}_{it} = 0$ or when $\bm{D}_{i,t+\delta} = 0$ with some $\delta > 0$.

\subsection{Target parameter and its causal interpretation} \label{subsec:ATEM}

For periods $t$ and $s$ such that $1 \le s < t$, a user-specified effective treatment function $E_t: \mathcal{D}^T \to \mathcal{E}_t$, and effective treatment intensity $e \in \tilde{\mathcal{E}}_t$, the ATEM is defined by 
\begin{align} \label{eq:ATEM}
	\mathrm{ATEM}(t, s, e)
	\coloneqq \bE[ Y_{it} - Y_{it}^*(0) \mid M_{i,t,s,e} = 1],
\end{align}
where $M_{i,t,s,e} \coloneqq \bm{1}\{ E_{it} = e, E_{is} = 0 \}$ indicates whether unit $i$ is a \textit{mover} who moves from the comparison group to receiving effective treatment intensity $e$ between periods $s$ and $t$.

The next theorem is our first main result, which shows that the causal interpretation of $\mathrm{ATEM}(t, s, e)$ is determined from the relationship between the pre-specified $E_t$ and true $E_t^*$.

\begin{theorem} \label{thm:interpretation}
	Suppose that Assumption \ref{as:effective} holds.
	\begin{enumerate}[(i)]
		\item If $E_t = E_t^*$,
		\begin{align*}
			\mathrm{ATEM}(t, s, e)
			= \mathrm{ATEM}^*(t, s, e)
			\coloneqq \bE[ Y_{it}^*(e) - Y_{it}^*(0) \mid M_{i,t,s,e}^* = 1 ],
		\end{align*}
		where $M_{i,t,s,e}^* \coloneqq \bm{1}\{ E_{it}^* = e, E_{is}^* = 0 \}$.
		\item If $E_t$ is not the true but a correct specification, there is a surjective $c_t: \mathcal{E}_t \to \mathcal{E}_t^*$ such that
		\begin{align*}
			\mathrm{ATEM}(t, s, e)
			= \bE[ Y_{it}^*(c_t(e)) - Y_{it}^*(0) \mid M_{i,t,s,e} = 1 ].
		\end{align*}
		\item Assume that $E_t^*$ has discrete realizations (for exposition purposes).
		If $E_t$ is incorrect,
		\begin{align*}
			\mathrm{ATEM}(t, s, e)
			= \sum_{e^* \in \mathcal{E}_t^*} \bE[ Y_{it}^*(e^*) - Y_{it}^*(0) \mid M_{i,t,s,e^*}^* = 1, M_{i,t,s,e} = 1 ] \cdot \bP(M_{i,t,s,e^*}^* = 1 \mid M_{i,t,s,e} = 1 ).
		\end{align*}
	\end{enumerate}
\end{theorem}

Result (i) considers the case where the pre-specified $E_t$ is actually true.
In this case, the ATEM indicates the ATE of receiving the true effective treatment intensity $e$ for movers defined with the true $E_t^*$.

Result (ii) corresponds to the situation where the pre-specified $E_t$ is correct but ``finer'' than the true $E_t^*$.
Then, the ATEM can be interpreted as the ATE of receiving the true effective treatment intensity $c_t(e)$ for movers in terms of the pre-specified $E_t$

Result (iii) allows for the possibility that the pre-specified $E_t$ is neither true nor correct.
Even in this case, the ATEM has causal interpretation as a weighted average of the ATE conditional on both the ``true mover'' ($M_{i,t,s,e^*}^* = 1$) and the ``pre-specified mover'' ($M_{i,t,s,e} = 1$).
The weight function is $\bP(M_{i,t,s,e^*}^* = 1 \mid M_{i,t,s,e} = 1)$, that is, the probability of being the true mover conditional on being the pre-specified mover.
Importantly, this weight function is non-negative and sums up to one so that the ATEM satisfies the so-called \textit{no-sign-reversal} property.
More precisely, $\mathrm{ATEM}(t, s, e)$ is positive (resp. negative) if the individual treatment effect $Y_{it}^*(e^*) - Y_{it}^*(0)$ is almost surely (a.s.) positive (resp. negative) for all $e^* \in \tilde{\mathcal{E}}_t^*$.
Thus, in conjunction with the identification results given in the next subsection, our methods do not suffer from the negative weighting or sign-reversal problem, unlike commonly used TWFE regressions (cf. \citealp{de2020two,de2022difference,de2022several}; \citealp{goodman2021difference}; \citealp{ishimaru2022what}).

\begin{remark}[Test of significance]
	If there is no treatment effect heterogeneity across units and true effective treatment intensity such that $Y_{it}^*(e^*) - Y_{it}^*(0) = \tau_t$ with some non-stochastic $\tau_t$, we have $\mathrm{ATEM}(t, s, e) = \tau_t$ for any $E_t$.
	In particular, this result holds for $\tau_t = 0$.
	Thus, in practice, we can test for the significance of treatment effects by performing statistical inference for $\mathrm{ATEM}(t, s, e)$ even when the pre-specified $E_t$ is incorrect.
\end{remark}

The choice of the effective treatment function $E_t$, periods $t > s$, and effective treatment intensity $e \in \tilde{\mathcal{E}}_t$ determine what $\mathrm{ATEM}(t, s, e)$ measures more specifically.

\begin{example}[continues = ex:once]
	Let $\mathrm{ATEM}^{\mathrm{O}}(t, s, 1) \coloneqq \bE[ Y_{it} - Y_{it}^*(0) \mid M_{i,t,s,1}^{\mathrm{O}} = 1 ]$, where $M_{i,t,s,1}^{\mathrm{O}} \coloneqq \bm{1}\{ E_{it}^{\mathrm{O}} = 1, E_{is}^{\mathrm{O}} = 0 \}$.
	By setting $s = 1$, we have
	\begin{align*}
		\mathrm{ATEM}^{\mathrm{O}}(t, 1, 1)
		= \bE[ Y_{it} - Y_{it}^*(0) \mid D_{i1} = 0, \bm{D}_{it} \neq 0 ],
	\end{align*}
	which captures the composition of the instantaneous and dynamic treatment effects, that is, the effect of receiving some treatment intensity at least once between periods $2$ and $t$ for those who did not receive the treatment in the first period.
\end{example}

\begin{example}[continues = ex:event]
	Let $\mathrm{ATEM}^{\mathrm{E}}(t, s, e) \coloneqq \bE[ Y_{it} - Y_{it}^*(0) \mid M_{i,t,s,e}^{\mathrm{E}} = 1 ]$, where $M_{i,t,s,e}^{\mathrm{E}} \coloneqq \bm{1}\{ E_{it}^{\mathrm{E}} = e, E_{is}^{\mathrm{E}} = 0 \}$ with $e \ge 2$ and $t \ge e$.
	Setting $s = e - 1$, we have
	\begin{align*}
		\mathrm{ATEM}^{\mathrm{E}}(t, e-1, e)
		= \bE \left[ Y_{it} - Y_{it}^*(0) \mid \bm{D}_{i,e-1} = 0, D_{ie} \neq 0 \right].
	\end{align*}
	This indicates the ATE in period $t$ for those who received some treatment intensity for the first time in period $e$.
	In the same manner as in the staggered case, estimating this ATEM over a set of $(t, e)$ allows us to understand the instantaneous and dynamic treatment effects separately.
	Specifically, $\mathrm{ATEM}^{\mathrm{E}}(t,  e-1, e)$ with $t = e$ (resp. with $t > e$) is informative about the instantaneous (resp. dynamic) effect of the first treatment receipt in period $e$.
\end{example}

\begin{example}[continues = ex:number]
	Define $\mathrm{ATEM}^{\mathrm{N}}(t, s, e) \coloneqq \bE[ Y_{it} - Y_{it}^*(0) \mid M_{i,t,s,e}^{\mathrm{N}} = 1 ]$, where $M_{i,t,s,e}^{\mathrm{N}} \coloneqq \bm{1}\{ E_{it}^{\mathrm{N}} = e, E_{is}^{\mathrm{N}} = 0 \}$ with $e \ge 1$ and $t \ge e + 1$.
	When we set $s = 1$, we have
	\begin{align*}
		\mathrm{ATEM}^{\mathrm{N}}(t, 1, e)
		= \bE \left[ Y_{it} - Y_{it}^*(0) \; \middle| \; D_{i1} = 0, \sum_{r=2}^t \bm{1}\{ D_{ir} \neq 0 \} = e \right].
	\end{align*}
	This ATEM recovers the ATE in period $t$ for those who were not treated in the first period and received treatment $e$ times in $t$ periods.
	By estimating this ATEM over a set of $(t, e)$, we can separately assess the instantaneous and dynamic effects of the number of treatment adoptions.
\end{example}

\begin{remark}[Relationship between different specifications]
	Interestingly, $\mathrm{ATEM}^{\mathrm{O}}(t, 1, 1)$ can be obtained from aggregating $\mathrm{ATEM}^{\mathrm{E}}(t, e-1, e)$ or $\mathrm{ATEM}^{\mathrm{N}}(t, 1, e)$ in a certain way.
	See the supplementary appendix for more details.
\end{remark}

\begin{remark}[Practical recommendations]
	It is common practice to perform event study analyses prior to formal DiD estimation (cf. \citealp{miller2023introductory}).
	Given this, it would be natural to take the event specification as a good starting point and to plot $\mathrm{ATEM}^{\mathrm{E}}(t, e-1, e)$ over a set of $(t, e)$, including ``pre-trends'' to assess the validity of parallel trends (see Figure \ref{fig:event} in the empirical illustration section).
	It would then be desirable to further examine treatment effect heterogeneity with other effective treatment functions, such as the number specification. 
	Finally, to obtain more precise estimates, it would be recommended to estimate $\mathrm{ATEM}^{\mathrm{O}}(t, 1, 1)$ and its time-series average as useful aggregation parameters, that is,
	\begin{align} \label{eq:aggregate}
		\mathrm{ATEM}^{\mathrm{A}}
		\coloneqq \frac{1}{(T - 1)} \sum_{t=2}^T \mathrm{ATEM}^{\mathrm{O}}(t, 1, 1).
	\end{align}
\end{remark}

\subsection{Identification of ATEM}

Let $S_{i,t,s} \coloneqq \bm{1}\{ E_{it} = 0, E_{is} = 0 \}$ indicate a \textit{stayer} who remains in the comparison group in periods $s$ and $t$.
Define the GPS by
\begin{align*}
	p_{t,s,e}(X_i) \coloneqq \bP(M_{i,t,s,e} = 1 \mid M_{i,t,s,e} + S_{i,t,s} = 1, X_i),
\end{align*}
which is the probability of being a mover with receiving effective treatment intensity $e \in \tilde{\mathcal{E}}_t$ between periods $s$ and $t$ conditional on $X_i$ and on being either a mover or a stayer between periods $s$ and $t$.
Let $\mathcal{A}$ denote a set of triplets $(t, s, e)$ at which we would like to identify $\mathrm{ATEM}(t, s, e)$.

\begin{assumption}[Overlap] \label{as:overlap}
	For each $(t, s, e) \in \mathcal{A}$, there exists a constant $\varepsilon > 0$ such that $\varepsilon \le \bE[ S_{i,t,s} \mid X_i ] \le 1 - \varepsilon$ and $\varepsilon \le p_{t,s,e}(X_i) \le 1 - \varepsilon$ a.s.
\end{assumption}

\begin{assumption}[Parallel trends] \label{as:parallel}
	For any $t \ge 2$, 
	\begin{align*}
		\bE[ Y_{it}^*(0) - Y_{i,t-1}^*(0) \mid \bm{D}_{iT}, X_i ]
		= \bE[ Y_{it}^*(0) - Y_{i,t-1}^*(0) \mid X_i]
		\quad \text{a.s.}
	\end{align*}
\end{assumption}

Assumption \ref{as:overlap} is an overlap condition under which there are non-negligible proportions of stayers and movers.
This is a common requirement, but clearly restricts the data generating process and functional form of effective treatment.
For instance, it will be violated if there are few treated units in each period or if $E_t$ is time invariant (i.e., $E_{it} = E_{is}$ for all $i, t, s$).

The parallel trends condition in Assumption \ref{as:parallel} is essential for our analysis.
This assumption states that the evolution of the untreated potential outcome is mean-independent of the entire treatment path, conditional on the unit-specific covariates.
Note that this is a condition on the data generating process and does not restrict the specification of effective treatment.
For a better understanding, consider the following specific model:
\begin{equation*}
	Y_{it}^*(0) = X_i^\top \gamma_t + \alpha_i + \eta_t + v_{it},
	\qquad 
	D_{it} = m(D_{i,t-1}, X_i, \alpha_i, \lambda_t, u_{it}),
\end{equation*}
where $\gamma_t$ is a non-stochastic coefficient vector, $m$ represents a treatment choice equation, $\alpha_i$ is a unit FE, $\eta_t$ and $\lambda_t$ are non-stochastic time FEs, and $v_{it}$ and $u_{it}$ are idiosyncratic error terms that vary across both $i$ and $t$.
The essential requirement is that $\alpha_i$ is additively separable from the other components in the untreated potential outcome equation.
If the treatment status in $t = 0$ is deterministic, then $\bm{D}_{iT}$ is determined from $(\alpha_i, \lambda_1, \dots, \lambda_T, u_{i1}, \dots, u_{iT})$ given $X_i$.
Then, Assumption \ref{as:parallel} is fulfilled if $(\alpha_i, u_{i1}, \dots, u_{iT})$ are independent of $(v_{i1}, \dots, v_{iT})$ conditional on $X_i$.
Importantly, this situation allows for the endogenous treatment choice in that $\alpha_i$ can be correlated with both $Y_{it}$ and $D_{it}$.

As the second main result in this paper, the next theorem shows that the ATEM is identifiable from each of OR, IPW, and DR estimands.
For a generic $A_{it}$ and periods $t > s$, denote $\Delta A_{i,t,s} \coloneqq A_{it} - A_{is}$.
Define
\begin{equation} \label{eq:estimand}
	\begin{split}
		\mathrm{ATEM}^{\mathrm{OR}}(t, s, e)
		& \coloneqq \bE \left[ w_{i,t,s,e}^{M} \big( \Delta Y_{i,t,s} - m_{t,s}(X_i) \big) \right], \\
		\mathrm{ATEM}^{\mathrm{IPW}}(t, s, e)
		& \coloneqq \bE \left[ \left( w_{i,t,s,e}^{M}  - w_{i,t,s,e}^{S} \right) \Delta Y_{i,t,s} \right], \\
		\mathrm{ATEM}^{\mathrm{DR}}(t, s, e)
		& \coloneqq \bE \left[ \left( w_{i,t,s,e}^{M}  - w_{i,t,s,e}^{S} \right) \big( \Delta Y_{i,t,s} - m_{t,s}(X_i) \big) \right],
	\end{split}
\end{equation}
where
\begin{align*}
	m_{t,s}(X_i) & \coloneqq \bE[ \Delta Y_{i,t,s} \mid S_{i,t,s} = 1, X_i ],
	& w_{i,t,s,e}^{M} & \coloneqq \frac{M_{i,t,s,e}}{\bE[M_{i,t,s,e}]}, \\
	r_{t,s,e}(X_i) & \coloneqq \frac{p_{t,s,e}(X_i)}{1 - p_{t,s,e}(X_i)},
	& w_{i,t,s,e}^{S} & \coloneqq \frac{ r_{t,s,e}(X_i) S_{i,t,s}}{\bE [ r_{t,s,e}(X_i) S_{i,t,s} ] }.
\end{align*}
In words, $m_{t,s}(X_i)$ is the OR function for stayers, $r_{t,s,e}(X_i)$ is the ratio of GPS, and $w_{i,t,s,e}^{M}$ and $w_{i,t,s,e}^{S}$ are weights for movers and stayers, respectively.

\begin{theorem} \label{thm:ATE}
	Suppose that Assumptions \ref{as:effective}--\ref{as:parallel} hold.
	For each $(t, s, e) \in \mathcal{A}$,
	\begin{align*}
		\mathrm{ATEM}(t, s, e)
		= \mathrm{ATEM}^{\mathrm{OR}}(t, s, e)
		= \mathrm{ATEM}^{\mathrm{IPW}}(t, s, e)
		= \mathrm{ATEM}^{\mathrm{DR}}(t, s, e).
	\end{align*}
\end{theorem}

\begin{remark}[Identification in the absence of covariates]
	If the identification conditions are fulfilled without covariates, the identification result reduces to
	\begin{align*}
		\mathrm{ATEM}(t, s, e) 
		= \bE[\Delta Y_{i,t,s} \mid M_{i,t,s,e} = 1] - \bE[\Delta Y_{i,t,s} \mid S_{i,t,s} = 1].
	\end{align*}
\end{remark}

In Theorem \ref{thm:ATE}, the three estimands serve the same role from the perspective of identification analysis, but we prefer the DR estimand for estimation purposes.
This is because only it has the DR property.
To see this, we introduce parametric models for the OR function and GPS, say $m_{t,s}(x; \gamma_{t,s})$ and $p_{t,s,e}(x; \pi_{t,s,e})$.
Here, $m_{t,s}(x; \gamma_{t,s})$ and $p_{t,s,e}(x; \pi_{t,s,e})$ are functions known up to the finite-dimensional parameters $\gamma_{t,s} \in \Gamma_{t,s}$ and $\pi_{t,s,e} \in \Pi_{t,s,e}$.
Let $\gamma_{t,s}^*$ and $\pi_{t,s,e}^*$ denote the pseudo-true values.
The DR estimand under these parametric specifications is given by
\begin{align} \label{eq:DRestimand}
	\mathrm{ATEM}^{\mathrm{DR}}(t, s, e; \gamma_{t,s}, \pi_{t,s,e})
	\coloneqq \bE \left[ \left( w_{i,t,s,e}^{M} - w_{i,t,s,e}^{S}(\pi_{t,s,e}) \right) \big( \Delta Y_{i,t,s} - m_{t,s}(X_i; \gamma_{t,s}) \big) \right],
\end{align}
where 
\begin{align} \label{eq:quantities2}
	r_{t,s,e}(X_i; \pi_{t,s,e}) \coloneqq \frac{p_{t,s,e}(X_i; \pi_{t,s,e})}{1 - p_{t,s,e}(X_i; \pi_{t,s,e})},
	\qquad 
	w_{i,t,s,e}^{S}(\pi_{t,s,e}) \coloneqq \frac{ r_{t,s,e}(X_i; \pi_{t,s,e}) S_{i,t,s}}{\bE [ r_{t,s,e}(X_i; \pi_{t,s,e}) S_{i,t,s} ] }.
\end{align}

\begin{assumption}[Parametric model] \label{as:correct}
	For each $(t, s, e) \in \mathcal{A}$, either condition is satisfied.
	\begin{enumerate}[(i)]
		\item There exists a unique $\gamma_{t,s}^* \in \Gamma_{t,s}$ such that $m_{t,s}(X_i) = m_{t,s}(X_i; \gamma_{t,s}^*)$ a.s.
		\item There exists a unique $\pi_{t,s,e}^* \in \Pi_{t,s,e}$ such that $p_{t,s,e}(X_i) = p_{t,s,e}(X_i; \pi_{t,s,e}^*)$ a.s.
	\end{enumerate}
\end{assumption}

\begin{theorem} \label{thm:DR}
	Suppose that Assumptions \ref{as:effective}--\ref{as:parallel} hold.
	Fix arbitrary $(t, s, e) \in \mathcal{A}$.
	\begin{enumerate}[(i)]
		\item Under Assumption \ref{as:correct}(i), $\mathrm{ATEM}(t, s, e) = \mathrm{ATEM}^{\mathrm{DR}}(t, s, e; \gamma_{t,s}^*, \pi_{t,s,e})$ for all $\pi_{t,s,e} \in \Pi_{t,s,e}$.
		\item Under Assumption \ref{as:correct}(ii), $\mathrm{ATEM}(t, s, e) = \mathrm{ATEM}^{\mathrm{DR}}(t, s, e; \gamma_{t,s}, \pi_{t,s,e}^*)$ for all $\gamma_{t,s} \in \Gamma_{t,s}$.
	\end{enumerate}
\end{theorem}

This result, in turn, ensures that the ATEM can be consistently estimated via the DR estimand if either the OR function or GPS is correctly specified.

\section{Estimation and Statistical Inference} \label{sec:estimation}

We consider estimation and inference methods based on the identification result using the DR estimand in \eqref{eq:DRestimand}.
We focus on the following data generating process:

\begin{assumption}[DGP] \label{as:DGP}
	The panel data $\{ (Y_{it}, D_{it}, X_i): i \in \mathcal{N}, t \in \mathcal{T} \}$ are independent and identically distributed (IID) across $i$.
\end{assumption}

\subsection{Estimation procedure} \label{subsec:procedure}

We propose a two-step estimation procedure.
First, we estimate the finite-dimensional parameters $\gamma_{t,s}$ and $\pi_{t,s,e}$ using parametric estimation methods.
For instance, $\gamma_{t,s}$ may be estimated by running a linear regression with the dependent variable $\Delta Y_{i,t,s}$ and the explanatory vector $X_i$ using the observations satisfying $S_{i,t,s} = 1$.
Meanwhile, using observations with $M_{i,t,s,e} + S_{i,t,s} = 1$, we can estimate $\pi_{t,s,e}$ using the parametric maximum likelihood method, such as logit or probit estimation, in which the binary response variable is $M_{i,t,s,e}$ and the explanatory vector is $X_i$.
Let $\hat \theta_{t,s,e} \coloneqq (\hat \gamma_{t,s}^\top, \hat \pi_{t,s,e}^\top)^\top$ denote the vector of the first-step estimators.

In the second step, we compute the sample analog of the DR estimand in \eqref{eq:DRestimand} as follows:
\begin{align} \label{eq:DRestimator}
	\widehat{\mathrm{ATEM}}(t, s, e)
	\coloneqq {\bE}_N \left[ \big( \hat w_{i,t,s,e}^{M}  - \hat w_{i,t,s,e}^{S}(\hat \pi_{t,s,e}) \big) \big( \Delta Y_{i,t,s} - m_{t,s}(X_i; \hat \gamma_{t,s}) \big) \right],
\end{align}
where ${\bE}_N[A_{it}] = N^{-1} \sum_{i=1}^N A_{it}$ for generic $A_{it}$ and
\begin{align*}
	\hat w_{i,t,s,e}^{M} 
	\coloneqq \frac{M_{i,t,s,e}}{\bE_N[M_{i,t,s,e}]},
	\qquad 
	\hat w_{i,t,s,e}^{S}(\hat \pi_{t,s,e})
	\coloneqq \frac{ r_{t,s,e}(X_i; \hat \pi_{t,s,e}) S_{i,t,s}}{\bE_N [ r_{t,s,e}(X_i; \hat \pi_{t,s,e}) S_{i,t,s} ] }.
\end{align*}

\begin{remark}[Overview of the asymptotic properties]
	When $N$ tends to infinity and $T$ is fixed, under a set of mild regularity conditions, we can show that the proposed ATEM estimator is $1/\sqrt{N}$-consistent and asymptotically normally distributed.
	More precisely, we can prove the following influence function representation:
	\begin{align*}
		\sqrt{N} \left( \widehat{\mathrm{ATEM}}(t, s, e) - \mathrm{ATEM}(t, s, e) \right)
		= \frac{1}{\sqrt{N}} \sum_{i=1}^N \psi_{i,t,s,e}( \theta_{t,s,e}^*) + o_{\bP}(1)
		\quad \text{for each $(t, s, e) \in \mathcal{A}$},
	\end{align*}
	where $\theta_{t,s,e}^* \coloneqq (\gamma_{t,s}^{*\top}, \pi_{t,s,e}^{*\top})^\top$ denotes the pseudo-true parameter vector in the first-stage parametric estimation and $\psi_{i,t,s,e}( \theta_{t,s,e}^*)$ is an influence function, which has a zero mean and a finite variance.
	See the supplementary appendix for the formal asymptotic analysis.
\end{remark}

\subsection{Multiplier bootstrap inference} \label{subsec:bootstrap}

Let $\hat \psi_{i,t,s,e}(\hat \theta_{t,s,e})$ be an estimator of $\psi_{i,t,s,e}(\theta_{t,s,e}^*)$ obtained by replacing the pseudo-true parameter value $\theta_{t,s,e}^*$ and population mean $\bE$ with the first-step estimator $\hat \theta_{t,s,e}$ and sample mean ${\bE}_N$, respectively.
Let $\{ V_i^{\star} \}_{i=1}^N$ be a set of IID bootstrap weights independent of the original sample such that $\bE[V_i^{\star}] = 0$, $\bE|V_i^{\star}|^2 = 1$, and $\bE|V_i^{\star}|^{2 + \varepsilon} < \infty$ for some $\varepsilon > 0$.
A common choice is \citeauthor{mammen1993bootstrap}'s (\citeyear{mammen1993bootstrap}) weight, such that $\bP( V_i^{\star} = 1 - \kappa ) = \kappa / \sqrt{5}$ and $\bP( V_i^{\star} = \kappa ) = 1 - \kappa / \sqrt{5}$ with $\kappa = (\sqrt{5} + 1) / 2$.
The bootstrap analog of the ATEM estimator is defined by
\begin{align*}
	\widehat{\mathrm{ATEM}}^{\star}(t, s, e)
	\coloneqq \widehat{\mathrm{ATEM}}(t, s, e) + {\bE}_N \left[ V_i^{\star} \cdot \hat \psi_{i,t,s,e}(\hat \theta_{t,s,e}) \right].
\end{align*}

We consider the following multiplier bootstrap inference procedure, which is similar to those of \cite{belloni2017program} and \cite{callaway2021difference}.

\begin{algorithm}[Multiplier bootstrap] \label{algorithm:bootstrap}
	\hfil 
	\begin{enumerate}
		\item Generate $\{ V_i^{\star} \}_{i=1}^N$.
		Compute $\widehat{\mathrm{ATEM}}^{\star}(t, s, e)$ and $\hat R^{\star}(t, s, e) \coloneqq \widehat{\mathrm{ATEM}}^{\star}(t, s, e) - \widehat{\mathrm{ATEM}}(t, s, e)$.
		Repeat this step $B$ times.
		\item Construct the bootstrap standard error for $\widehat{\mathrm{ATEM}}(t, s, e)$ by the empirical interquartile range of $\hat R^{\star}(t, s, e)$ rescaled with the interquartile range of $\mathrm{Normal}(0, 1)$:
		\begin{align*}
			\mathrm{SE}(t, s, e) 
			\coloneqq \frac{q_{0.75}(t, s, e) - q_{0.25}(t, s, e)}{z_{0.75} - z_{0.25}},
		\end{align*}
		where $q_{\alpha}(t, s, e)$ is the empirical $\alpha$-quantile of $\hat R^{\star}(t, s, e)$ and $z_{\alpha}$ is the $\alpha$-quantile of $\mathrm{Normal}(0, 1)$.
		\item Construct the $1 - \alpha$ UCB for $\mathrm{ATEM}(t, s, e)$ as follows:
		\begin{align*}
			\widehat{\mathrm{CI}}_{1-\alpha}(t, s, e)
			\coloneqq \left[ \widehat{\mathrm{ATEM}}(t, s, e) - \hat c_{1-\alpha} \cdot \mathrm{SE}(t, s, e), \quad \widehat{\mathrm{ATEM}}(t, s, e) + \hat c_{1-\alpha} \cdot \mathrm{SE}(t, s, e) \right],
		\end{align*}
		where the bootstrap critical value $\hat c_{1 - \alpha}$ is obtained by the empirical $(1 - \alpha)$-quantile of the maximum of the bootstrapped $t$ test statistics defined as $\text{max-t-stat}_{\mathcal{A}}^{\star} \coloneqq \max_{(t, s, e) \in \mathcal{A}} |\hat R^{\star}(t, s, e)| /  \mathrm{SE}(t, s, e)$.
	\end{enumerate}
\end{algorithm}

It can be shown that the bootstrap distribution mimics the asymptotic distribution of the ATEM estimator.
The proof is analogous to Theorem 3 of \cite{callaway2021difference} and thus is omitted here. 
This result, in turn, ensures the asymptotic validity of the multiplier bootstrap inference procedure.
More precisely, the coverage error of the $1 - \alpha$ UCB vanishes asymptotically:
\begin{align*}
	\lim_{N \to \infty} \bP \left( \mathrm{ATEM}(t, s, e) \in \widehat{\mathrm{CI}}_{1-\alpha}(t, s, e) \quad \text{for all $(t, s, e) \in \mathcal{A}$} \right) = 1 - \alpha.
\end{align*}

\section{Testable Implication for Parallel Trends: Pre-trends Test} \label{sec:testable}

To begin with, we add a mild requirement to the pre-specified $E_t$.
The following assumption states that a comparison unit in the current period should be a comparison unit even in past periods.

\begin{assumption}[Effective treatment 2] \label{as:nondecreasing}
	For any $t \ge 2$, $E_{it} = 0$ implies $E_{i,t-1} = 0$.
\end{assumption}

The basic idea is as follows.
For a period $r$ such that $2 \le r \le s < t$, we denote $\Delta Y_{ir} \coloneqq Y_{ir} - Y_{i,r-1}$.
Under the parallel trends condition in Assumption \ref{as:parallel}, we can show that
\begin{equation} \label{eq:pre-trends}
	\begin{split}
		\bE[ \Delta Y_{ir} \mid M_{i,t,s,e} = 1, X_i]
		& = \bE[ \Delta Y_{ir} \mid S_{i,t,s} = 1, X_i]
		\quad \text{a.s.}
	\end{split}
\end{equation}
Thus, it is necessary for parallel trends that, in any ``pre-effective treatment'' period $r$, the outcome evolution for movers should follow the same path as the outcome evolution for stayers.
If we find statistical evidence against \eqref{eq:pre-trends}, the parallel trends assumption would be implausible.

The next theorem formalizes this idea using the DR estimand.
For each $(t, s, e) \in \mathcal{A}$ and $r \ge 2$, define
\begin{align*}
	\mathrm{ATEM}(t, s, r, e)
	& \coloneqq \bE[ Y_{ir} - Y_{ir}^*(0) \mid M_{i,t,s,e} = 1]
\end{align*}	
and
\begin{align*}
	\mathrm{ATEM}^{\mathrm{DR}}(t, s, r, e)
	& \coloneqq \bE \left[ \left( w_{i,t,s,e}^{M} - w_{i,t,s,e}^{S} \right) \big( \Delta Y_{ir} - m_{t,s,r}(X_i) \big) \right],
\end{align*}
where $m_{t,s,r}(X_i) \coloneqq \bE[ \Delta Y_{ir} \mid S_{i,t,s} = 1, X_i ]$ is the OR function in period $r$ for stayers between periods $s$ and $t$.
Note that $\mathrm{ATEM}(t, s, r, e)$ for $t = r$ reduces to $\mathrm{ATEM}(t, s, e)$ in \eqref{eq:ATEM}.

\begin{theorem} \label{thm:testable}
	Suppose that Assumptions \ref{as:effective}, \ref{as:overlap}, and \ref{as:nondecreasing} hold.
	Under Assumption \ref{as:parallel},
	\begin{align} \label{eq:testable}
		\mathrm{ATEM}(t, s, r, e) 
		= \mathrm{ATEM}^{\mathrm{DR}}(t, s, r, e) 
		= 0
	\end{align}
	for all $(t, s, e) \in \mathcal{A}$ and $r$ such that $2 \le r \le s < t$.
\end{theorem}

This testable implication suggests the following procedure for assessing parallel trends.
First, we construct the $1 - \alpha$ UCB for $\mathrm{ATEM}(t, s, r, e)$ using the multiplier bootstrap.
Then, the parallel trends assumption should be refuted if a number of the resulting confidence bands exclude zeros.

It should be cautious that ``passing'' this type of test does not necessarily ensure the validity of parallel trends.
This is because the testable implication in \eqref{eq:testable} is merely a necessary condition for parallel trends.
Indeed, recent studies point out that this type of test often has low power against the violation of parallel trends (cf., \citealp{roth2022pretest}).

\section{Empirical Illustration: Union Wage Premium} \label{sec:application}

\cite{vella1998whose} find a positive effect of union membership status on wages by estimating a dynamic selection model.
In the DiD literature, \cite{de2020two} revisit the same dataset and demonstrate that current union membership may not significantly influence current wages.
Here, we complement these previous results with the proposed methods, which enable the estimation of both the instantaneous and dynamic treatment effects.

The data record 545 full-time working males taken from the National Longitudinal Survey (Youth Sample) for 1980--1987.\footnote{
	The data set is available through the Inter-university Consortium for Political and Social Research (\citealp{de2020data}).
}
The binary treatment $D_{it}$ indicates the union membership status of individual $i$ in year $t$, reflecting whether $i$ is covered by a collective bargaining agreement.
The outcome $Y_{it}$ is the logarithm of the hourly wages for individual $i$ in year $t$.
As the set of unit-specific covariates $X_i$, we consider race, years of schooling, and years of labor market experience in the first period.

Figures \ref{fig:once}, \ref{fig:event}, and \ref{fig:number} present the estimates of the three ATEM parameters in Examples \ref{ex:once}, \ref{ex:event}, and \ref{ex:number} with the corresponding 95\% UCBs obtained from 5,000 bootstrap replications using \citeauthor{mammen1993bootstrap}'s (\citeyear{mammen1993bootstrap}) weight.
For the first-step estimation, we estimate the OR function for stayers and GPS using the method of least squares and logit estimation, respectively.
We can see that there are both positive and negative estimates, but many of them are small in the absolute sense.
Indeed, all the resulting UCBs include zeros, that is, we find no statistical evidence for current and future union wage premiums.
We also find a statistically insignificant estimate of the aggregation parameter $\mathrm{ATEM}^{\mathrm{A}}$ defined in \eqref{eq:aggregate}.
Specifically, the estimate is $0.041$ and the 95\% confidence interval is $[-0.076, 0.159]$.

The pre-trends testing result for the event specification is shown in Figure \ref{fig:event}.
No UCBs for ``pre-effective treatment'' periods exclude zeros, which is consistent with the parallel trends in Assumption \ref{as:parallel}.

Accordingly, our DiD analysis suggests that union membership might not substantially improve current and future wages, which is in line with the finding of \cite{de2020two}.

\clearpage

\bibliography{ref.bib}

\clearpage

\begin{figure}[p]
	\centering
	\includegraphics[width = \textwidth, bb = 0 0 720 720]{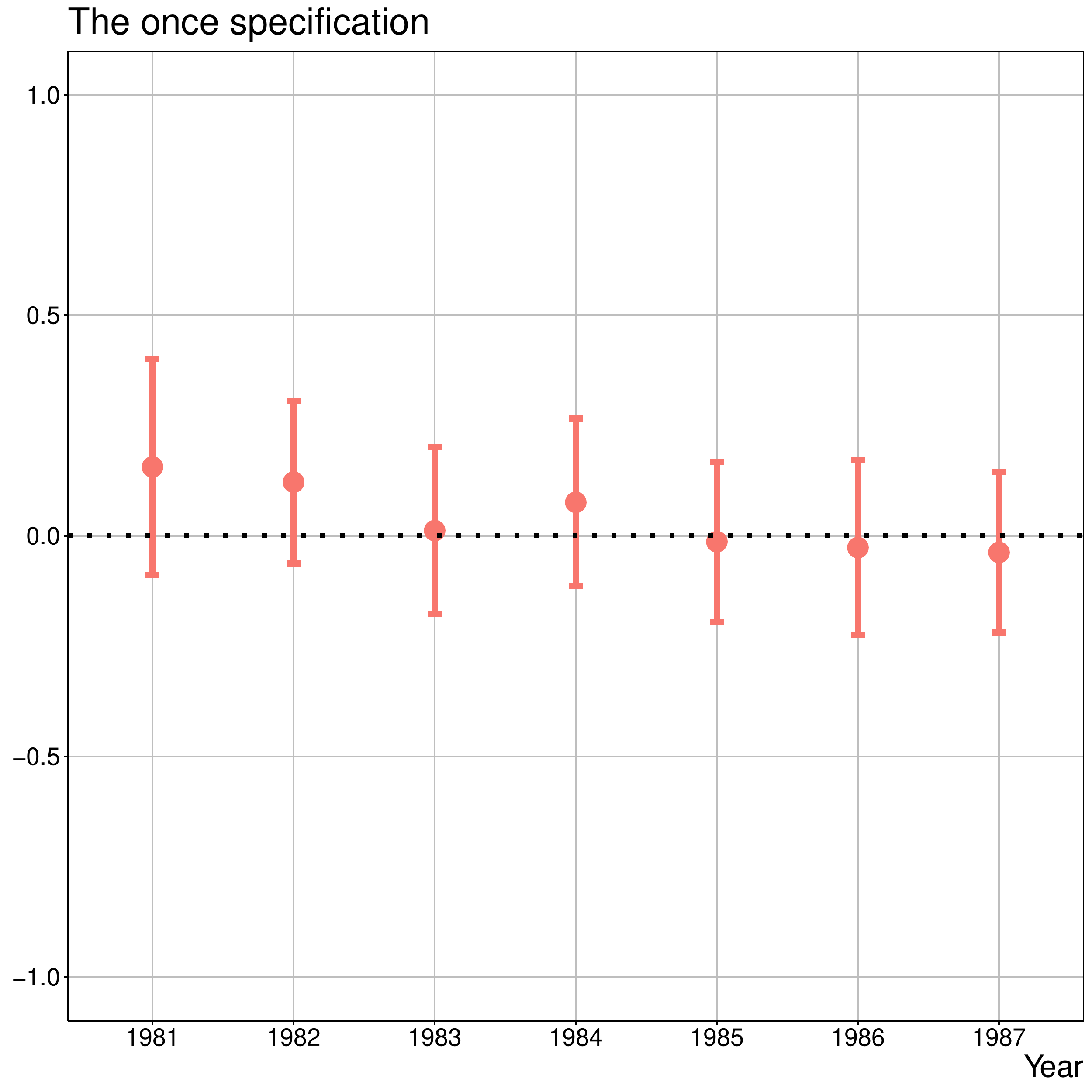}
	\caption{The empirical results: $\mathrm{ATEM}^{\mathrm{O}}(t, 1, 1) = \bE[ Y_{it} - Y_{it}^*(0) \mid M_{i,t,1,e}^{\mathrm{O}} = 1]$}
	\label{fig:once}
	\bigskip 
	\begin{flushleft}
		Note: Circles and bars indicate ATEM estimates and corresponding 95\% UCBs, respectively.
	\end{flushleft}
\end{figure}

\clearpage

\begin{figure}[p]
	\centering
	\includegraphics[width = \textwidth, bb = 0 0 1440 1440]{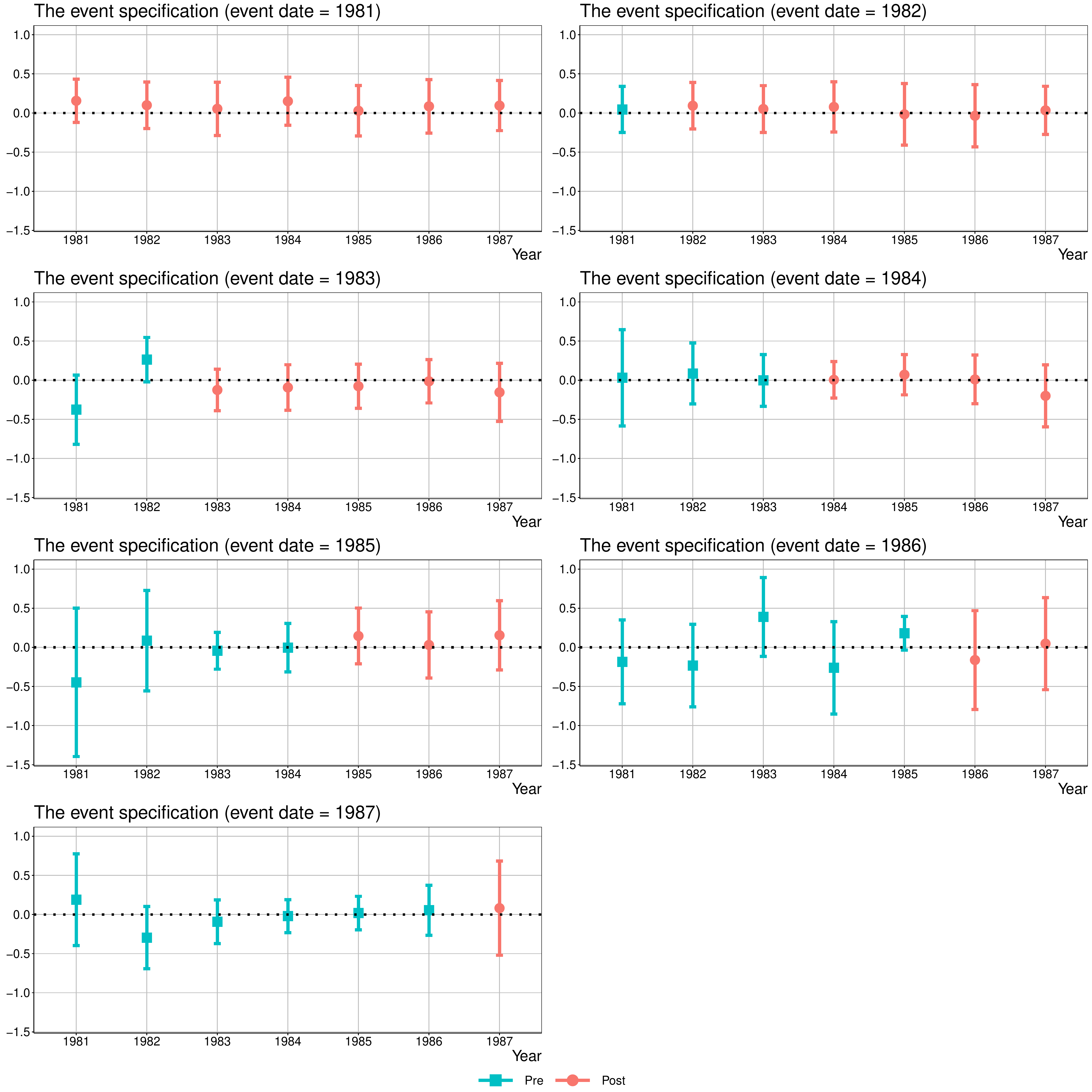}
	\caption{The empirical results: $\mathrm{ATEM}^{\mathrm{E}}(t, e-1, e) = \bE[ Y_{it} - Y_{it}^*(0) \mid M_{i,t,e-1,e}^{\mathrm{E}} = 1]$ and $\mathrm{ATEM}^{\mathrm{E}}(t, e-1, r, e) = \bE[ Y_{ir} - Y_{ir}^*(0) \mid M_{i,t,e-1,e}^{\mathrm{E}} = 1]$}
	\label{fig:event}
	\bigskip 
	\begin{flushleft}
		Note: Circles (resp. boxes) and bars indicate ATEM estimates and corresponding 95\% UCBs for ``post-effective treatment'' (resp. ``pre-effective treatment'') periods.
	\end{flushleft}
\end{figure}

\clearpage

\begin{figure}[p]
	\centering
	\includegraphics[width = \textwidth, bb = 0 0 1440 1440]{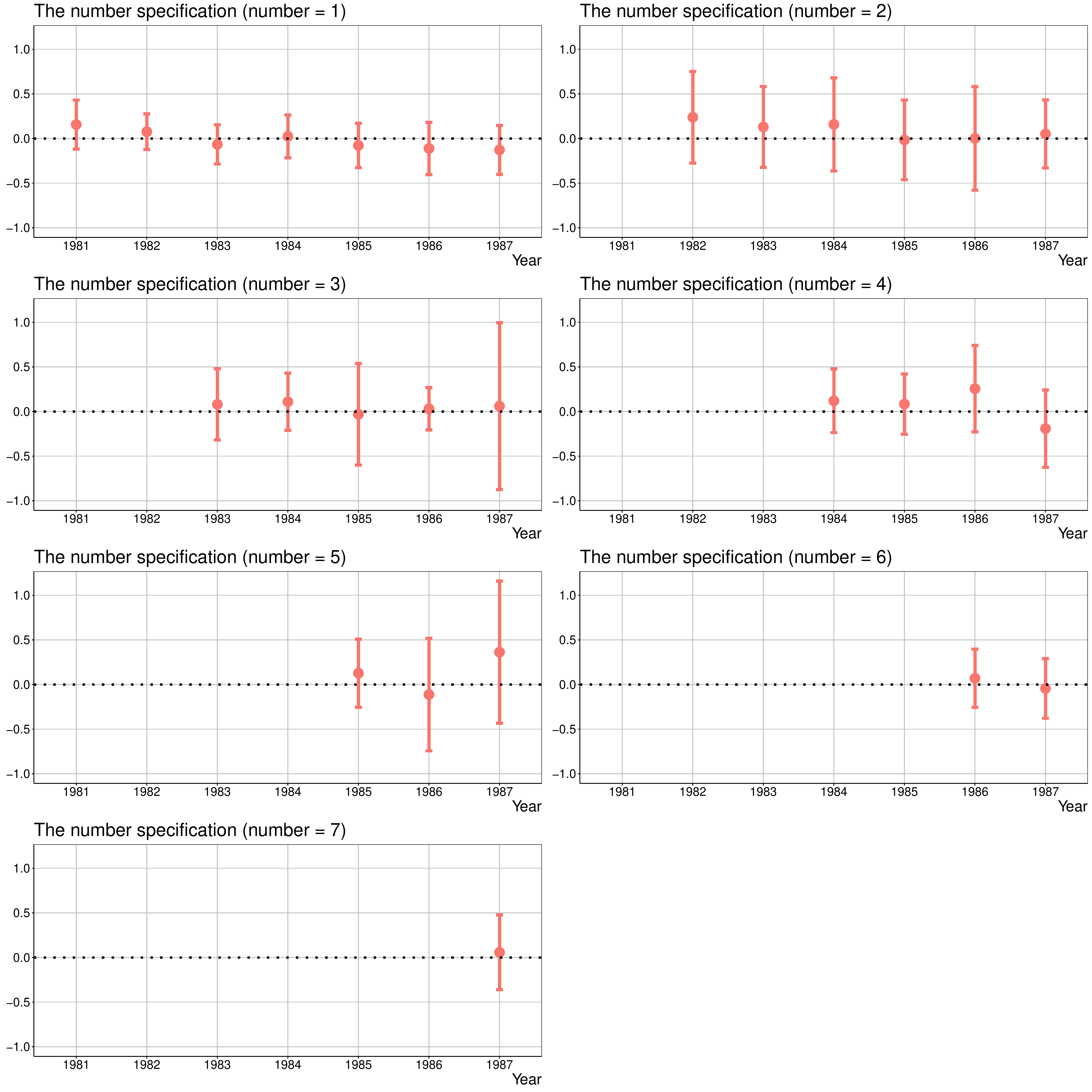}
	\caption{The empirical results: $\mathrm{ATEM}^{\mathrm{N}}(t, 1, e) = \bE[ Y_{it} - Y_{it}^*(0) \mid M_{i,t,1,e}^{\mathrm{N}} = 1]$}
	\label{fig:number}
	\bigskip 
	\begin{flushleft}
		Note: Circles and bars indicate ATEM estimates and corresponding 95\% UCBs, respectively.
	\end{flushleft}
\end{figure}

\clearpage 

\begin{center}
	{\Large Online Appendix for: \\
	An Effective Treatment Approach to Difference-in-Differences with General Treatment Patterns}
	
	\bigskip
	
	{\large Takahide Yanagi$^\dagger$}
	\bigskip
	
	$^\dagger$ Graduate School of Economics, Kyoto University.
\end{center}

\appendix

\renewcommand{\thepage}{S\arabic{page}}
\renewcommand{\thetable}{S\arabic{table}}
\renewcommand{\thefigure}{S\arabic{figure}}
\setcounter{page}{1}
\setcounter{table}{0}
\setcounter{figure}{0}

\begin{abstract}
	Appendices \ref{sec:proof} and \ref{sec:lemma} contain the proofs of the theorems and several lemmas.
	The discussion for aggregation parameters and the formal asymptotic investigations are given in Appendix \ref{sec:additional}.
	In Appendix \ref{sec:simulation}, we report the results of Monte Carlo simulations.
\end{abstract}

\section{Proofs} \label{sec:proof}

\subsection{Proof of Theorem \ref{thm:interpretation}}

To begin with, recall the definition of the ATEM parameter:
\begin{align*}
	\mathrm{ATEM}(t, s, e)
	= \bE[ Y_{it} - Y_{it}^*(0) \mid M_{i,t,s,e} = 1]
	\quad
	\text{with $M_{i,t,s,e} = \bm{1}\{ E_{it} = e, E_{is} = 0 \}$.} 
\end{align*}

\paragraph{Result (i).}

When the pre-specified $E_t$ is actually true so that $E_t = E_t^*$, we have $M_{i,t,s,e} = M_{i,t,s,e}^* = \bm{1}\{ E_{it}^* = e, E_{is}^* = 0 \}$.
Since $Y_{it} = Y_{it}^*(e)$ conditional on $E_{it}^* = e$, it holds that
\begin{align*}
	\mathrm{ATEM}(t, s, e)
	= \bE [ Y_{it}^*(e) - Y_{it}^*(0) \mid M_{i,t,s,e}^* = 1 ]
	= \mathrm{ATEM}^*(t, s, e).
\end{align*}

\paragraph{Result (ii).}

When the pre-specified $E_t: \mathcal{D}^T \to \mathcal{E}_t$ is not true but correct, the true $E_t^*: \mathcal{D}^T \to \mathcal{E}_t^*$ is coarser than $E_t$, and there is a surjective mapping $c_t: \mathcal{E}_t \to \mathcal{E}_t^*$ such that $E_{it}^* = c_t(E_{it})$.
Noting that $Y_{it} = Y_{it}^*(E_{it}^*)$, we have
\begin{align*}
	\mathrm{ATEM}(t, s, e)
	& = \bE [ Y_{it}^*(E_{it}^*) - Y_{it}^*(0) \mid M_{i,t,s,e} = 1 ] \\
	& = \bE [ Y_{it}^*(c_t(E_{it})) - Y_{it}^*(0) \mid M_{i,t,s,e} = 1 ] \\
	& = \bE [ Y_{it}^*(c_t(e)) - Y_{it}^*(0) \mid M_{i,t,s,e} = 1 ].
\end{align*}

\paragraph{Result (iii).}

Using the law of iterated expectations, we have
\begin{align*}
	& \mathrm{ATEM}(t, s, e) \\
	& = \bE[ Y_{it} - Y_{it}^*(0) \mid E_{it} = e, E_{is} = 0] \\
	& = \bE \Big[ \bE[ Y_{it} - Y_{it}^*(0) \mid E_{it}^*, E_{is}^*, E_{it} = e, E_{is} = 0 ] \; \Big| \; E_{it} = e, E_{is} = 0 \Big] \\
	& = \sum_{e^* \in \mathcal{E}_t^*} \bE[ Y_{it}^*(e^*) - Y_{it}^*(0) \mid E_{it}^* = e^*, E_{is}^* = 0, E_{it} = e, E_{is} = 0 ] \cdot \bP( E_{it}^* = e^*, E_{is}^* = 0 \mid E_{it} = e, E_{is} = 0 ) \\
	& = \sum_{e^* \in \mathcal{E}_t^*} \bE[ Y_{it}^*(e^*) - Y_{it}^*(0) \mid M_{i,t,s,e^*}^* = 1, M_{i,t,s,e} = 1 ] \cdot \bP( M_{i,t,s,e^*}^* = 1 \mid M_{i,t,s,e} = 1 ),
\end{align*}
where we used Assumption \ref{as:effective}(ii) and the definitions of $M_{i,t,s,e^*}^* = \bm{1}\{ E_{it}^* = e^*, E_{is}^* = 0 \}$ and $M_{i,t,s,e} = \bm{1}\{ E_{it} = e, E_{is} = 0 \}$ in the third and last equalities, respectively.

\qed

\subsection{Proof of Theorem \ref{thm:ATE}}

\paragraph{The OR estimand.}
Let $\mathrm{ATEM}_{X_i}(t, s, e) \coloneqq \bE[ Y_{it} - Y_{it}^*(0) \mid M_{i,t,s,e} = 1, X_i ]$.
Using the law of iterated expectations and Lemma \ref{lem:ATE}, we have
\begin{align*}
	\mathrm{ATEM}(t, s, e)
	& = \bE[ \mathrm{ATEM}_{X_i}(t, s, e) \mid M_{i,t,s,e} = 1 ] \\
	& = \bE \Big[ \bE[ \Delta Y_{i,t,s} \mid M_{i,t,s,e} = 1, X_i ] - \bE[ \Delta Y_{i,t,s} \mid S_{i,t,s} = 1, X_i ] \; \Big| \; M_{i,t,s,e} = 1 \Big] \\
	& = \bE[ \Delta Y_{i,t,s} - m_{t,s}(X_i) \mid M_{i,t,s,e} = 1 ] \\
	& = \bE \left[ w_{i,t,s,e}^{M} \big( \Delta Y_{i,t,s} - m_{t,s}(X_i) \big) \right] \\
	& = \mathrm{ATEM}^{\mathrm{OR}}(t, s, e).
\end{align*}

\paragraph{The IPW estimand.}
Given the first result, we can obtain the desired result if we show that $\bE [ w_{i,t,s,e}^{S} \Delta Y_{i,t,s} ] = \bE [ w_{i,t,s,e}^{M} \cdot m_{t,s}(X_i) ]$.
First, it is straightforward to see that
\begin{equation} \label{eq:propensity}
	\begin{split}
		p_{t,s,e}(X_i) 
		& = \frac{\bP( M_{i,t,s,e} = 1 \mid X_i )}{\bP(M_{i,t,s,e} = 1 \lor S_{i,t,s} = 1 \mid X_i)}
		= \frac{\bE[M_{i,t,s,e} \mid X_i]}{\bE[M_{i,t,s,e} \mid X_i] + \bE[S_{i,t,s} \mid X_i]},\\
		1 - p_{t,s,e}(X_i) 
		& = \frac{\bE[S_{i,t,s} \mid X_i]}{\bE[M_{i,t,s,e} \mid X_i] + \bE[S_{i,t,s} \mid X_i]}, \\
		r_{t,s,e}(X_i) 
		& = \frac{ p_{t,s,e}(X_i) }{ 1 - p_{t,s,e}(X_i) }
		= \frac{\bE[M_{i,t,s,e} \mid X_i]}{\bE[S_{i,t,s} \mid X_i]}.
	\end{split}
\end{equation}
Then, the law of iterated expectations leads to
\begin{align} \label{eq:taumean}
	\bE [ r_{t,s,e}(X_i) S_{i,t,s} ]
	= \bE \left[ \frac{\bE[M_{i,t,s,e} \mid X_i]}{\bE[S_{i,t,s} \mid X_i]} S_{i,t,s} \right]
	= \bE[M_{i,t,s,e}].
\end{align}
Thus, using \eqref{eq:propensity}, \eqref{eq:taumean}, and the law of iterated expectations, we have
\begin{align*}
	\bE[ w_{i,t,s,e}^{S} \Delta Y_{i,t,s} ]
	& = \bE\left[ \frac{r_{t,s,e}(X_i) S_{i,t,s}}{\bE[ r_{t,s,e}(X_i) S_{i,t,s} ]} \Delta Y_{i,t,s} \right] \\
	& = \bE\left[ \frac{\bE[M_{i,t,s,e} \mid X_i]}{ \bE[ M_{i,t,s,e} ]} \frac{S_{i,t,s}}{\bE[S_{i,t,s} \mid X_i]} \Delta Y_{i,t,s} \right] \\
	& = \bE\left[ \frac{ \bE[M_{i,t,s,e} \mid X_i] }{ \bE[ M_{i,t,s,e} ] } \bE[\Delta Y_{i,t,s} \mid S_{i,t,s} = 1, X_i ] \right] \\
	& = \bE\left[ w_{i,t,s,e}^{M} \cdot m_{t,s}(X_i) \right].
\end{align*}

\paragraph{The DR estimand.}
Using \eqref{eq:taumean}, we have
\begin{align*}
	\bE \left[ \left( w_{i,t,s,e}^{M} - w_{i,t,s,e}^{S} \right) m_{t,s}(X_i) \right]
	& = \frac{1}{ \bE[ M_{i,t,s,e} ] } \bE \left[ \big( M_{i,t,s,e} - r_{t,s,e}(X_i) S_{i,t,s} \big) m_{t,s}(X_i) \right] \\
	& = \frac{1}{\bE[M_{i,t,s,e}]} \bE \left[ \big( \bE[M_{i,t,s,e} \mid X_i] - r_{t,s,e}(X_i) \bE[ S_{i,t,s} \mid X_i ] \big) m_{t,s}(X_i) \right] \\
	& = \frac{1}{\bE[M_{i,t,s,e}]} \bE \left[ \big( \bE[M_{i,t,s,e} \mid X_i] - \bE[M_{i,t,s,e} \mid X_i] \big) m_{t,s}(X_i) \right] \\
	& = 0,
\end{align*}
where the last two lines follow from the law of iterated expectations and \eqref{eq:propensity}.
Thus, in conjunction with the second result, we obtain the desired result.
\qed 

\subsection{Proof of Theorem \ref{thm:DR}}

We focus on proving result (i) only, because result (ii) can be shown by the same arguments as in the proof of the third result of Theorem \ref{thm:ATE}.
Owing to the first result in Theorem \ref{thm:ATE}, it suffices to show that
\begin{align*}
	\bE \left[ w_{i,t,s,e}^{S}(\pi_{t,s,e}) \left( \Delta Y_{i,t,s} - m_{t,s}(X_i; \gamma_{t,s}^*) \right) \right]
	= 0
	\quad \text{for all $\pi_{t,s,e} \in \Pi_{t,s,e}$.}
\end{align*}
Fix arbitrary $\pi_{t,s,e} \in \Pi_{t,s,e}$. 
Using the law of iterated expectations, the left-hand side can be written as
\begin{align*}
	\bE \left[ \frac{ r_{t,s,e}(X_i; \pi_{t,s,e}) }{\bE [ r_{t,s,e}(X_i; \pi_{t,s,e}) S_{i,t,s} ] } \bE \left[ S_{i,t,s} \big( \Delta Y_{i,t,s} - m_{t,s}(X_i; \gamma_{t,s}^*) \big) \; \middle| \; X_i \right] \right].
\end{align*}
From the definition of $S_{i,t,s}$ and the correctly specified $m_{t,s}(\cdot) = m_{t,s}(\cdot; \gamma_{t,s}^*)$, it is easy to see that 
\begin{align*}
	\bE \left[ S_{i,t,s} \big( \Delta Y_{i,t,s} - m_{t,s}(X_i; \gamma_{t,s}^*) \big) \; \middle| \; X_i \right] 
	= \bE[ S_{i,t,s} \mid X_i ] \cdot m_{t,s}(X_i) - \bE[ S_{i,t,s} \mid X_i] \cdot m_{t,s}(X_i)
	= 0.
\end{align*}
This completes the proof.
\qed

\subsection{Proof of Theorem \ref{thm:testable}}

The same arguments as in the proof of Theorem \ref{thm:ATE} can show that
\begin{align*}
	\mathrm{ATEM}(t, s, r, e) 
	= \mathrm{ATEM}^{\mathrm{DR}}(t, s, r, e).
\end{align*}
Furthermore, by Assumptions \ref{as:effective}(ii) and \ref{as:nondecreasing}, it is straightforward to see that 
\begin{align*}
	\mathrm{ATEM}(t, s, r, e) 
	= \bE[ Y_{ir}^*(0) - Y_{ir}^*(0) \mid M_{i,t,s,e} = 1 ]
	= 0,
\end{align*}
which completes the proof.
\qed 

\section{Lemmas} \label{sec:lemma}

\begin{lemma} \label{lem:parallel}
	Under Assumption \ref{as:parallel},
	\begin{align*}
		\bE [ Y_{ir}^*(0) - Y_{i,r-1}^*(0) \mid M_{i,t,s,e}, X_i ] 
		= \bE[ Y_{ir}^*(0) - Y_{i,r-1}^*(0) \mid S_{i,t,s}, X_i ] 
		\quad \text{a.s.,}
	\end{align*}
	for any $e \in \tilde{\mathcal{E}}_t$ and $t, s, r \in \mathcal{T}$ such that $t > s$ and $r \ge 2$.
\end{lemma}

\begin{proof}
	Using the law of iterated expectations, we have
	\begin{align*}
		\bE[ Y_{ir}^*(0) - Y_{i,r-1}^*(0) \mid M_{i,t,s,e}, X_i ]
		& = \bE \big[ \bE[ Y_{ir}^*(0) - Y_{i,r-1}^*(0) \mid \bm{D}_{iT}, X_i ] \; \big| \; M_{i,t,s,e}, X_i \big] \\
		& = \bE[ Y_{ir}^*(0) - Y_{i,r-1}^*(0) \mid X_i]
		\quad \text{a.s.,}
	\end{align*}
	where the first equality follows from the fact that $M_{i,t,s,e} = \bm{1}\{ E_{it} = e, E_{is} = 0 \}$ is determined from $\bm{D}_{iT}$ and the second equality follows from Assumption \ref{as:parallel}.
	The same argument shows that
	\begin{align*}
		\bE[ Y_{ir}^*(0) - Y_{i,r-1}^*(0) \mid S_{i,t,s}, X_i ] = \bE[ Y_{ir}^*(0) - Y_{i,r-1}^*(0) \mid X_i]
		\quad \text{a.s.}
	\end{align*}
\end{proof}

Let
\begin{align*}
	\mathrm{ATEM}_x(t, s, e)
	\coloneqq \bE[ Y_{it} - Y_{it}^*(0) \mid M_{i,t,s,e} = 1, X_i = x ].
\end{align*}

\begin{lemma} \label{lem:ATE}
	Under Assumptions \ref{as:effective}--\ref{as:parallel}, $\mathrm{ATEM}_x(t, s, e)$ for each $(t, s, e, x) \in \mathcal{A} \times \mathcal{X}$ is identified by
	\begin{align*}
		\mathrm{ATEM}_x(t, s, e) 
		= \bE[ \Delta Y_{i,t,s} \mid M_{i,t,s,e} = 1, X_i = x ] - \bE[ \Delta Y_{i,t,s} \mid S_{i,t,s} = 1, X_i = x ].
	\end{align*}
\end{lemma}

\begin{proof}
	Under Assumption \ref{as:overlap}, we have
	\begin{align*}
		& \bE[ \Delta Y_{i,t,s} \mid M_{i,t,s,e} = 1, X_i = x ] - \bE[ \Delta Y_{i,t,s} \mid S_{i,t,s} = 1, X_i = x ] \\
		& = \bE[ Y_{it} - Y_{is}^*(E_{is}^*) \mid M_{i,t,s,e} = 1, X_i = x ] - \bE[ Y_{it}^*(E_{it}^*) - Y_{is}^*(E_{is}^*) \mid S_{i,t,s} = 1, X_i = x ]\\
		& = \bE[ Y_{it} - Y_{is}^*(0) \mid M_{i,t,s,e} = 1, X_i = x ] - \bE[ Y_{it}^*(0) - Y_{is}^*(0) \mid S_{i,t,s} = 1, X_i = x ]\\
		& = \bE[ Y_{it} - Y_{it}^*(0) \mid M_{i,t,s,e} = 1, X_i = x ] \\
		& \quad + \bE[ Y_{it}^*(0) - Y_{is}^*(0) \mid M_{i,t,s,e} = 1, X_i = x ] - \bE[ Y_{it}^*(0) - Y_{is}^*(0) \mid S_{i,t,s} = 1, X_i = x ] \\
		& = \bE[ Y_{it} - Y_{it}^*(0) \mid M_{i,t,s,e} = 1, X_i = x ] \\
		& \quad + \sum_{r=s+1}^{t} \Big( \bE[ Y_{ir}^*(0) - Y_{i,r-1}^*(0) \mid M_{i,t,s,e} = 1, X_i = x ] - \bE[ Y_{ir}^*(0) - Y_{i,r-1}^*(0) \mid S_{i,t,s} = 1, X_i = x ] \Big) \\
		& = \mathrm{ATEM}_x(t, s, e),
	\end{align*}
	where we used Assumption \ref{as:effective}(ii) and Lemma \ref{lem:parallel} in the second and last equalities, respectively.
\end{proof}

\section{Additional Results} \label{sec:additional}

\subsection{Aggregation parameters} \label{subsec:aggregation}

Given Theorems \ref{thm:ATE} and \ref{thm:DR}, a number of aggregation parameters can be recovered as weighted averages of $\mathrm{ATEM}(t, s, e)$.
We discuss the specific aggregation schemes by continuing Examples \ref{ex:event} and \ref{ex:number}.

\begin{example}[continues = ex:event]
	Define
	\begin{align} \label{eq:aggregate1}
		\mathrm{ATEM}^{\mathrm{E}}(t)
		& \coloneqq \bE[ Y_{it} - Y_{it}^*(0) \mid M_{i,t,1}^{\mathrm{E}} = 1 ],
	\end{align}
	where $M_{i,t,1}^{\mathrm{E}} \coloneqq \bm{1}\{ E_{it}^{\mathrm{E}} \neq 0, E_{i1}^{\mathrm{E}} = 0 \}$.
	This is the ATE for those who were not treated in the first period but received some treatment intensity between periods 2 and $t$.
	From the law of iterated expectations and the fact that $E_{it}^{\mathrm{E}} = e$ implies that $E_{i,e-1}^{\mathrm{E}} = 0$, we can observe that
	\begin{align*}
		\mathrm{ATEM}^{\mathrm{E}}(t) 
		& = \bE \big[ \bE[ Y_{it} - Y_{it}^*(0) \mid E_{it}^{\mathrm{E}}, M_{i,t,1}^{\mathrm{E}} = 1 ] \; \big| \; M_{i,t,1}^{\mathrm{E}} = 1 \big] \\
		& = \sum_{e=2}^t \bE[ Y_{it} - Y_{it}^*(0) \mid M_{i,t,1,e}^{\mathrm{E}} = 1 ] \cdot \bP( M_{i,t,1,e}^{\mathrm{E}} = 1 \mid M_{i,t,1}^{\mathrm{E}} = 1 ) \\
		& = \sum_{e = 2}^t \mathrm{ATEM}^{\mathrm{E}}(t, e-1, e) \cdot \omega^{\mathrm{E}}(t, e-1, e),
	\end{align*}
	where $\omega^{\mathrm{E}}(t, e - 1, e) \coloneqq \bE[ M_{i,t,e - 1,e}^{\mathrm{E}} ] / \sum_{\tilde e=2}^t \bE[ M_{i,t,\tilde e-1,\tilde e}^{\mathrm{E}} ]$.
	Note that the weight function is identifiable from data and sums up to one.
\end{example}

\begin{example}[continues = ex:number]
	Define
	\begin{align} \label{eq:aggregate2}
		\mathrm{ATEM}^{\mathrm{N}}(t, 1) 
		\coloneqq \bE[ Y_{it} - Y_{it}^*(0) \mid M_{i,t,1}^{\mathrm{N}} = 1 ],
	\end{align}
	where $M_{i,t,1}^{\mathrm{N}} \coloneqq \bm{1}\{ E_{it}^{\mathrm{N}} \neq 0, E_{i1}^{\mathrm{N}} = 0 \}$ indicates whether unit $i$, who did not receive the treatment in the first period, received some treatment intensity until period $t$.
	Using the law of iterated expectations, we can see that this aggregation parameter can be recovered from
	\begin{align*}
		\mathrm{ATEM}^{\mathrm{N}}(t, 1)
		& = \bE \left[ \bE[ Y_{it} - Y_{it}^*(0) \mid E_{it}^{\mathrm{N}}, M_{i,t,1}^{\mathrm{N}} = 1 ] \mid M_{i,t,1}^{\mathrm{N}} = 1 \right] \\
		& = \sum_{e \in \tilde{\mathcal{E}}_t^{\mathrm{N}}} \mathrm{ATEM}^{\mathrm{N}}(t, 1, e) \cdot \bP( M_{i,t,1,e}^{\mathrm{N}} = 1 \mid M_{i,t,1}^{\mathrm{N}} = 1 ) \\
		& = \sum_{e \in \tilde{\mathcal{E}}_t^{\mathrm{N}}} \mathrm{ATEM}^{\mathrm{N}}(t, 1, e) \cdot \omega^{\mathrm{N}}(t,1,e),
	\end{align*}
	where $\tilde{\mathcal{E}}_t^{\mathrm{N}}$ denotes the set of non-zero realizations of $E_{it}^{\mathrm{N}}$ and $\omega^{\mathrm{N}}(t, s, e) \coloneqq \bE[ M_{i,t,s,e}^{\mathrm{N}} ] / \sum_{\tilde e \in \tilde{\mathcal{E}}_t^{\mathrm{N}}} \bE[ M_{i,t,s,\tilde e}^{\mathrm{N}} ]$ is an identifiable weight function that is non-negative and sums up to one.
\end{example}

The aggregation parameters in \eqref{eq:aggregate1} and \eqref{eq:aggregate2} are closely related to the ATEM parameter discussed in Example \ref{ex:once}.
Specifically, it is straightforward to see that
\begin{align*}
	\mathrm{ATEM}^{\mathrm{O}}(t, 1, 1)
	= \mathrm{ATEM}^{\mathrm{E}}(t)
	= \mathrm{ATEM}^{\mathrm{N}}(t, 1)
	\quad \text{for any $t \ge 2$.}
\end{align*}
Thus, we recommend that researchers obtain sensible aggregation parameters by estimating $\mathrm{ATEM}^{\mathrm{O}}(t, 1, 1)$ and its time-series average, that is,
\begin{align*}
	\mathrm{ATEM}^{\mathrm{A}}
	\coloneqq \frac{1}{T-1} \sum_{t=2}^T \mathrm{ATEM}^{\mathrm{O}}(t, 1, 1).
\end{align*}
We can perform the statistical inference for $\mathrm{ATEM}^{\mathrm{A}}$ by the multiplier bootstrap in a similar manner to Procedure \ref{algorithm:bootstrap}. 

\subsection{Asymptotic investigations} \label{subsec:asymptotic}

We study the asymptotic properties of the proposed ATEM estimator in \eqref{eq:DRestimator}.
We consider the asymptotic regime where $N$ grows to infinity and $T$ is fixed.

The following assumption comprises mild regularity conditions for the first-step estimation, which is similar to Assumption 7 of \cite{callaway2021difference}.
Let $g(x; \vartheta)$ denote a generic notation that indicates $m_{t,s}(x; \gamma_{t,s})$ or $p_{t,s,e}(x; \pi_{t,s,e})$ with $\vartheta = \gamma_{t,s}$ or $\vartheta = \pi_{t,s,e}$, respectively.
We write the Frobenius norm of a generic matrix $A$ as $\| A \| = \sqrt{\mathrm{tr}(A^{\top} A)}$.

\begin{assumption}[First-step estimation] \label{as:asymptotic}
	The following conditions are satisfied for each $(t, s, e) \in \mathcal{A}$.
	\begin{enumerate}[(i)]
		\item The parameter space of $\vartheta$, $\Theta$, is a compact set on $\mathbb{R}^{\dim(\vartheta)}$.
		\item $g(X_i; \vartheta)$ is a.s. continuous in $\vartheta$ on $\Theta$.
		\item There exists a unique pseudo-true parameter value $\vartheta^* \in \mathrm{int}(\Theta)$.
		\item There exists a neighborhood of $\vartheta^*$, $\Theta^* \subset \Theta$, on which $g(X_i; \vartheta)$ is a.s. twice continuously differentiable in $\vartheta$ with bounded derivatives.
		\item The first-step estimator $\hat \vartheta$ satisfies $\hat \vartheta \overset{a.s.}{\to} \vartheta^*$ and
		\begin{align*}
			\sqrt{N} (\hat \vartheta - \vartheta^*)
			= \frac{1}{\sqrt{N}} \sum_{i=1}^N \psi_i^{\vartheta}(\vartheta^*) + o_{\bP}(1),
		\end{align*}
		where $\psi_i^{\vartheta}(\vartheta^*)$ is a $\dim(\vartheta) \times 1$ vector of influence functions such that $\bE[\psi_i^{\vartheta}(\vartheta^*)] = 0$, $\bE[\psi_i^{\vartheta}(\vartheta^*) \psi_i^{\vartheta}(\vartheta^*)^\top]$ is positive definite, and $\lim_{s \to 0} \bE [ \sup_{\vartheta \in \Theta^*: \| \vartheta - \vartheta^* \| \le s} \| \psi_i^{\vartheta}(\vartheta) - \psi_i^{\vartheta}(\vartheta^*) \|^2 ] = 0$.
		\item There exists some $\varepsilon > 0$ such that for all $\pi_{t,s,e} \in \Pi_{t,s,e}$, $\varepsilon \le p_{t,s,e}(X_i; \pi_{t,s,e}) \le 1 - \varepsilon$ a.s.
	\end{enumerate}
\end{assumption}

To state the next assumption, we introduce the following notation.
Letting $\theta_{t,s,e} = (\gamma_{t,s}^\top, \pi_{t,s,e}^\top)^\top$ be a vector of the first-step parameters, define
\begin{align} \label{eq:influence}
	\psi_{i,t,s,e}(\theta_{t,s,e})
	\coloneqq \psi_{i,t,s,e}^{M}(\gamma_{t,s}) - \psi_{i,t,s,e}^{S}(\theta_{t,s,e}) - \psi_{i,t,s,e}^{\mathrm{est}}(\theta_{t,s,e}),
\end{align}
where
\begin{align*}
	\psi_{i,t,s,e}^{M}(\gamma_{t,s})
	& \coloneqq w_{i,t,s,e}^{M} (\Delta Y_{i,t,s} - m_{t,s}(X_i; \gamma_{t,s})) - w_{i,t,s,e}^{M} \cdot \bE \left[ w_{i,t,s,e}^{M} \big( \Delta Y_{i,t,s} - m_{t,s}(X_i; \gamma_{t,s}) \big) \right], \\
	\psi_{i,t,s,e}^{S}(\theta_{t,s,e})
	& \coloneqq w_{i,t,s,e}^{S}(\pi_{t,s,e}) \cdot \big( \Delta Y_{i,t,s} - m_{t,s}(X_i; \gamma_{t,s}) \big) \\
	& \quad - w_{i,t,s,e}^{S}(\pi_{t,s,e}) \cdot \bE \left[ w_{i,t,s,e}^{S}(\pi_{t,s,e}) \cdot \big( \Delta Y_{i,t,s} - m_{t,s}(X_i; \gamma_{t,s}) \big) \right], \\
	\psi_{i,t,s,e}^{\mathrm{est}}(\theta_{t,s,e})
	& \coloneqq \left[ \psi_{i,t,s}^{\gamma}(\gamma_{t,s}) \right]^\top M_{t,s,e}^{1}(\theta_{t,s,e}) + \left[ \psi_{i,t,s,e}^{\pi}(\pi_{t,s,e}) \right]^\top M_{t,s,e}^{2}(\theta_{t,s,e}),
\end{align*}
with
\begin{align*}
	M_{t,s,e}^{1}(\theta_{t,s,e})
	& \coloneqq \bE \left[ \big( w_{i,t,s,e}^{M} - w_{i,t,s,e}^{S}(\pi_{t,s,e}) \big) \cdot \dot{m}_{t,s}(X_i; \gamma_{t,s}) \right], \\
	M_{t,s,e}^{2}(\theta_{t,s,e})
	& \coloneqq \bE \left[ \dot{w}_{i,t,s,e}^{S}(\pi_{t,s,e}) \cdot \big( \Delta Y_{i,t,s} - m_{t,s}(X_i; \gamma_{t,s}) \big) \right] \\
	& \quad - \bE[ \dot{w}_{i,t,s,e}^{S}(\pi_{t,s,e}) ] \cdot \bE \left[ w_{i,t,s,e}^{S}(\pi_{t,s,e}) \cdot \big( \Delta Y_{i,t,s} - m_{t,s}(X_i; \gamma_{t,s}) \big) \right],
\end{align*}
and
\begin{align*}
	\dot{m}_{t,s}(X_i; \gamma_{t,s})
	& \coloneqq \frac{\partial m_{t,s}(X_i; \gamma_{t,s})}{\partial \gamma_{t,s}},
	& \dot{p}_{t,s,e}(X_i; \pi_{t,s,e})
	& \coloneqq \frac{\partial p_{t,s,e}(X_i; \pi_{t,s,e})}{\partial \pi_{t,s,e}}, \\
	\dot{w}_{i,t,s,e}^{S}(\pi_{t,s,e})
	& \coloneqq \frac{ \dot{r}_{t,s,e}(X_i; \pi_{t,s,e}) S_{i,t,s} }{ \bE [ r_{t,s,e}(X_i; \pi_{t,s,e}) S_{i,t,s} ] },
	& \dot{r}_{t,s,e}(X_i; \pi_{t,s,e})
	& \coloneqq \frac{\dot{p}_{t,s,e}(X_i; \pi_{t,s,e})}{[ 1 - p_{t,s,e}(X_i; \pi_{t,s,e}) ]^2}.
\end{align*}
Here, it is easy to see that $\bE[ \psi_{i,t,s,e}(\theta_{t,s,e}^*) ] = 0$.
Also note that $\psi_{i,t,s,e}^{\mathrm{est}}$ arises from the fact that we estimate the unknown $\theta_{t,s,e}$ in the first-stage parametric estimation.

\begin{assumption}[Moments] \label{as:moment}
	\hfil 
	\begin{enumerate}[(i)]
		\item $\bE[ \Delta Y_{i,t,s}^2 ] < \infty$ for each $(t, s)$ such that $(t, s, e) \in \mathcal{A}$ with some $e \in \tilde{\mathcal{E}}_t$.
		\item $\bE[ (\psi_{i,t,s,e}(\theta_{t,s,e}^*))^2 ] < \infty$ for each $(t, s, e) \in \mathcal{A}$.
	\end{enumerate}
\end{assumption}

The next theorem presents the joint asymptotic normality result for $\widehat{\mathrm{ATEM}}(t, s, e)$ defined in \eqref{eq:DRestimator}.
The proof is relegated to the end of this section.
Let $\mathrm{ATEM}_{\mathcal{A}}$ and $\widehat{\mathrm{ATEM}}_{\mathcal{A}}$ denote the vectors stacking $\mathrm{ATEM}(t, s, e)$ and $\widehat{\mathrm{ATEM}}(t, s, e)$, respectively, over $(t, s, e) \in \mathcal{A}$.
Similarly, for each $i \in \mathcal{N}$, let $\psi_{i, \mathcal{A}}(\theta^*)$ be a vector that collects $\psi_{i,t,s,e}(\theta_{t,s,e}^*)$ over $(t, s, e) \in \mathcal{A}$.

\begin{theorem} \label{thm:asymptotic}
	Suppose that Assumptions \ref{as:effective}--\ref{as:correct}, \ref{as:DGP}, and \ref{as:asymptotic}--\ref{as:moment} hold.
	When $N \to \infty$, we have the following influence function representation:
	\begin{align*}
		\sqrt{N} \left( \widehat{\mathrm{ATEM}}(t, s, e) - \mathrm{ATEM}(t, s, e) \right)
		= \frac{1}{\sqrt{N}} \sum_{i=1}^N \psi_{i,t,s,e}( \theta_{t,s,e}^*) + o_{\bP}(1)
		\quad \text{for each $(t, s, e) \in \mathcal{A}$}.
	\end{align*}
	As a result, we have
	\begin{align*}
		\sqrt{N} \left( \widehat{\mathrm{ATEM}}_{\mathcal{A}} - \mathrm{ATEM}_{\mathcal{A}} \right)
		\overset{d}{\to}
		\mathrm{Normal} \left( 0, \Sigma_{\mathcal{A}} \right),
	\end{align*}
	where $\Sigma_{\mathcal{A}} \coloneqq \bE[ \psi_{i,\mathcal{A}}(\theta^*) \psi_{i,\mathcal{A}}(\theta^*)^\top ]$ denotes the asymptotic covariance matrix.
\end{theorem}

This theorem implies that the asymptotic variance of $\hat{\mathrm{ATEM}}(t, s, e)$ is given by $\bE[(\psi_{i,t,s,e}(\theta_{t,s,e}^*))^2] / N$, which depends critically on the inverse of $\bE[M_{i,t,s,e}]$.
Thus, a more precise ATEM estimation can be generally achieved with a coarser effective treatment specification, such as the once specification, as a coarser one can generally induce more movers between two periods.

\subsection{Proof of Theorem \ref{thm:asymptotic}}

Observe that 
\begin{align*}
	\widehat{\mathrm{ATEM}}(t, s, e) - \mathrm{ATEM}(t, s, e)
	& = \left( {\bE}_N[ \hat w_{i,t,s,e}^{M} \cdot \Delta Y_{i,t,s} ] - {\bE}[ w_{i,t,s,e}^{M} \cdot \Delta Y_{i,t,s} ] \right) \\
	& \quad - \left( {\bE}_N[ \hat w_{i,t,s,e}^{S}(\hat \pi_{t,s,e}) \cdot \Delta Y_{i,t,s} ] - {\bE}[ w_{i,t,s,e}^{S}(\pi_{t,s,e}^*) \cdot \Delta Y_{i,t,s} ] \right) \\
	& \quad - \left( {\bE}_N[ \hat w_{i,t,s,e}^{M} \cdot m_{t,s}(X_i; \hat \gamma_{t,s}) ] - {\bE}[ w_{i,t,s,e}^{M} \cdot m_{t,s}(X_i; \gamma_{t,s}^*) ] \right) \\
	& \quad + \left( {\bE}_N[ \hat w_{i,t,s,e}^{S}(\hat \pi_{t,s,e}) \cdot m_{t,s}(X_i; \hat \gamma_{t,s}) ] - {\bE}[ w_{i,t,s,e}^{S}(\pi_{t,s,e}^*) \cdot m_{t,s}(X_i; \gamma_{t,s}^*) ] \right) \\
	& \eqqcolon (\hat A_1 - A_1) - (\hat A_2 - A_2) - (\hat A_3 - A_3) + (\hat A_4 - A_4).
\end{align*}

For $\hat A_1$, we first note that 
\begin{align*}
	& {\bE}_N[M_{i,t,s,e}] - {\bE}[M_{i,t,s,e}] = O_{\bP}(N^{-1/2}), \\
	& {\bE}_N \left[ \frac{M_{i,t,s,e}}{{\bE}[M_{i,t,s,e}]^2} \Delta Y_{i,t,s} \right] - \bE \left[ \frac{M_{i,t,s,e}}{{\bE}[M_{i,t,s,e}]^2} \Delta Y_{i,t,s} \right] = O_{\bP}(N^{-1/2}),
\end{align*}
by the Lindeberg--L\'evy central limit theorem (CLT).
Applying Taylor's theorem to $1 / {\bE}_N[M_{i,t,s,e}]$ in $\hat A_1$ around $\bE[M_{i,t,s,e}]$ and using the above equations, we have
\begin{align*}
	\hat A_1
	& = {\bE}_N \left[ \frac{M_{i,t,s,e}}{{\bE}_N[M_{i,t,s,e}]} \Delta Y_{i,t,s} \right] \\
	& = {\bE}_N \left[ \frac{M_{i,t,s,e}}{{\bE}[M_{i,t,s,e}]} \Delta Y_{i,t,s} \right] \\
	& \quad - ({\bE}_N[M_{i,t,s,e}] - {\bE}[M_{i,t,s,e}]) \cdot {\bE}_N \left[ \frac{M_{i,t,s,e}}{{\bE}[M_{i,t,s,e}]^2} \Delta Y_{i,t,s} \right] + O_{\bP} \left( ({\bE}_N[M_{i,t,s,e}] - {\bE}[M_{i,t,s,e}])^2 \right) \\
	& = {\bE}_N \left[ w_{i,t,s,e}^{M} \Delta Y_{i,t,s} \right] - \frac{ {\bE}_N[M_{i,t,s,e}] - {\bE}[M_{i,t,s,e}] }{ \bE[M_{i,t,s,e}] } {\bE} \left[ w_{i,t,s,e}^{M} \Delta Y_{i,t,s} \right] + O_{\bP}(N^{-1}) \\
	& = A_1 + {\bE}_N \left[ w_{i,t,s,e}^{M} \Delta Y_{i,t,s} \right] - {\bE}_N \left[ w_{i,t,s,e}^{M} \right] \bE \left[ w_{i,t,s,e}^{M} \Delta Y_{i,t,s} \right] + O_{\bP}(N^{-1}),
\end{align*}
where we used Assumptions \ref{as:overlap} and \ref{as:DGP} in the second equality to evaluate the remainder term.
Thus, 
\begin{align*}
	\hat A_1 - A_1
	= {\bE}_N \left[ w_{i,t,s,e}^{M} \left( \Delta Y_{i,t,s} - \bE[w_{i,t,s,e}^{M} \Delta Y_{i,t,s}] \right) \right] + O_{\bP}(N^{-1}).
\end{align*}

Similarly, we can obtain the following results:
\begin{align*}
	\hat A_2 - A_2
	& = {\bE}_N \left[ w_{i,t,s,e}^{S}(\pi_{t,s,e}^*) \left( \Delta Y_{i,t,s} - \bE\left[ w_{i,t,s,e}^{S}(\pi_{t,s,e}^*) \Delta Y_{i,t,s} \right] \right) \right] \\
	& \quad + {\bE}_N\left[ \psi_{i,t,s,e}^{\pi}(\pi_{t,s,e}^*) \right]^\top \bE \left[ \dot{w}_{i,t,s,e}^{S}(\pi_{t,s,e}^*) \left( \Delta Y_{i,t,s} - \bE\left[ w_{i,t,s,e}^{S}(\pi_{t,s,e}^*) \Delta Y_{i,t,s} \right] \right) \right] + O_{\bP}(N^{-1}),
\end{align*}
and
\begin{align*}
	\hat A_3 - A_3
	& = {\bE}_N \left[ w_{i,t,s,e}^{M} \left( m_{t,s}(X_i; \gamma_{t,s}^*) - \bE \left[ w_{i,t,s,e}^{M} m_{t,s}(X_i; \gamma_{t,s}^*) \right] \right) \right] \\
	& \quad + {\bE}_N \left[ \psi_{t,s}^{\gamma}(\gamma_{t,s}^*) \right]^\top \bE \left[ w_{i,t,s,e}^{M} \dot{m}_{t,s}(X_i; \gamma_{t,s}^*) \right] + O_{\bP}(N^{-1}),
\end{align*}
and
\begin{align*}
	\hat A_4 - A_4
	& = {\bE}_N \left[ w_{i,t,s,e}^{S}(\pi_{t,s,e}^*) \left( m_{t,s}(X_i; \gamma_{t,s}^*) - \bE \left[ w_{i,t,s,e}^{S}(\pi_{t,s,e}^*) m_{t,s}(X_i; \gamma_{t,s}^*) \right] \right) \right] \\
	& \quad + {\bE}_N \left[ \psi_{i,t,s,e}^{\pi}(\pi_{t,s,e}^*) \right]^\top \bE \left[ \dot{w}_{i,t,s,e}^{S}(\pi_{t,s,e}^*) \left( m_{t,s}(X_i; \gamma_{t,s}^*) + \bE \left[ w_{i,t,s,e}^{S}(\pi_{t,s,e}^*) m_{t,s}(X_i; \gamma_{t,s}^*) \right] \right) \right] \\
	& \quad + {\bE}_N \left[ \psi_{t,s}^{\gamma}(\gamma_{t,s}^*) \right]^\top \bE \left[ w_{i,t,s,e}^{S}(\pi_{t,s,e}^*) \dot{m}_{t,s}(X_i; \gamma_{t,s}^*) \right] + O_{\bP}(N^{-1}).
\end{align*}

Consequently, we obtain
\begin{align*}
	& \widehat{\mathrm{ATEM}}(t, s, e) - \mathrm{ATEM}(t, s, e) \\
	& = {\bE}_N \left[ w_{i,t,s,e}^{M} \left( \Delta Y_{i,t,s} - m_{t,s}(X_i; \gamma_{t,s}^*) - \bE \left[ w_{i,t,s,e}^{M} \left( \Delta Y_{i,t,s} - m_{t,s}(X_i; \gamma_{t,s}^*) \right) \right] \right) \right] \\
	& \quad - {\bE}_N \left[ w_{i,t,s,e}^{S}(\pi_{t,s,e}^*) \left( \Delta Y_{i,t,s} - m_{t,s}(X_i; \gamma_{t,s}^*) - \bE \left[ w_{i,t,s,e}^{S}(\pi_{t,s,e}^*) \left( \Delta Y_{i,t,s} - m_{t,s}(X_i; \gamma_{t,s}^*) \right) \right] \right) \right] \\
	& \quad - {\bE}_N \left[ \psi_{t,s}^{\gamma}(\gamma_{t,s}^*) \right]^\top \bE \left[ \left( w_{i,t,s,e}^{M} - w_{i,t,s,e}^{S}(\pi_{t,s,e}^*) \right) \dot{m}_{t,s}(X_i; \gamma_{t,s}^*) \right] \\
	& \quad - {\bE}_N \left[ \psi_{i,t,s,e}^{\pi}(\pi_{t,s,e}^*) \right]^\top \bE \left[ \dot{w}_{i,t,s,e}^{S}(\pi_{t,s,e}^*) \left( \Delta Y_{i,t,s} - m_{t,s}(X_i; \gamma_{t,s}^*) - \bE[ w_{i,t,s,e}^{S}(\pi_{t,s,e}^*) (\Delta Y_{i,t,s} - m_{t,s}(X_i; \gamma_{t,s}^*))] \right) \right] \\
	& \quad + O_{\bP}(N^{-1}) \\
	& = {\bE}_N \left[ \psi_{i,t,s,e}^{M}(\gamma_{t,s}^*) - \psi_{i,t,s,e}^{S}(\theta_{t,s,e}^*) - \psi_{i,t,s,e}^{\mathrm{est}}(\theta_{t,s,e}^*) \right] + O_{\bP}(N^{-1}) \\
	& = {\bE}_N [\psi_{i,t,s,e}(\theta_{t,s,e}^*)] + O_{\bP}(N^{-1}),
\end{align*}
which completes the proof of asymptotic linear representation.
Then, using the CLT, it is easy to obtain the joint asymptotic normality result for $\sqrt{N}(\widehat{\mathrm{ATEM}}_{\mathcal{A}} - \mathrm{ATEM}_{\mathcal{A}})$.
\qed 

\section{Monte Carlo Experiment} \label{sec:simulation}

We evaluate the finite-sample performance of the proposed methods through Monte Carlo simulations.

Setting $D_{i1} = 0$ for all $i$, we generate a time-varying binary treatment $D_{it} = \bm{1} \{ \pi_1 + X_i \pi_2 + \alpha_i + \lambda_t \ge u_{it} \}$ for each $t \ge 2$, where the unit-specific covariate $X_i$, unit FE $\alpha_i$, and idiosyncratic error term $u_{it}$ are mutually independent standard normal variables, and $\lambda_t = t / T$ is the non-stochastic time effect.
We set the untreated potential outcome equation as $Y_{it}^*(0) = X_i \gamma_t + \alpha_i + \eta_t + v_{it}$ with the standard normal $v_{it}$ and another time effect $\eta_t = t$.
The observed outcome is generated by $Y_{it} = Y_{it}^*(0) + \sum_{e=2}^t \tau_{t,e} \cdot \bm{1}\{ E_{it}^{\mathrm{E}} = e \} + \xi_{it}$, where $\xi_{it} \sim \mathrm{Normal}(0, 1)$.
The true effective treatment is the event specification in \eqref{eq:event}, and the true ATEM corresponds to $\mathrm{ATEM}^{\mathrm{E}}(t, e-1, e) = \tau_{t,e}$.
The coefficient parameters are set to $\pi_1 = -1$, $\pi_2 = 1$, $\gamma_t = t$, and $\tau_{t,e} = (t + T - e) / T$.
We consider three patterns of sample sizes $(N, T) \in \{ (250, 4), (1000, 4), (4000, 4) \}$.

In the first-step estimation, we estimate the OR function $m_{t,s}(X_i)$ by regressing $\Delta Y_{i,t,s}$ on the constant and $X_i$ using observations satisfying $S_{i,t,s} = 1$.
In this simulation setup, $m_{t,s}(X_i) = \gamma_{1,t,s} + X_i \gamma_{2,t,s}$ with $\gamma_{1,t,s} = \eta_t - \eta_s$ and $\gamma_{2,t,s} = \gamma_t - \gamma_s$; thus, the OR function is correctly specified.
Meanwhile, we estimate the GPS using logit estimation, implying that the specification of GPS is incorrect.
In the second-step estimation, we estimate the three ATEM parameters considered in the empirical illustration.
For the multiplier bootstrap inference, we use \citeauthor{mammen1993bootstrap}'s (\citeyear{mammen1993bootstrap}) weight, and the number of bootstrap replications is set to 5,000.

Tables \ref{table:once}, \ref{table:event}, and \ref{table:number} present the simulation results obtained with 10,000 simulation repetitions. 
Each table reports the bias, root mean squared error (RMSE), point-wise coverage probability (PW.CP) and uniform coverage probability (U.CP) of the 95\% UCB, and average length of confidence interval (CI.L) for each effective treatment specification.

The bias is satisfactorily small even when $N$ is small.
The RMSE is approximately halved as $N$ increases from 250 to 1,000 or from 1,000 to 4,000, as expected from the $1/\sqrt{N}$ consistency of the ATEM estimator.
The PW.CP tends to be larger than the nominal level of 0.95, whereas the U.CP is closer to the desired level, which is a common feature of the uniform inference.
As expected, the CI.L becomes shorter as $N$ increases.
Overall, the RMSE and U.CP for the once specification are more desirable than those for the others, due to the fact that the once specification can induce more movers than the others.

\clearpage

\begin{table}[p]
	\caption{Monte Carlo simulation results: The once specification $\mathrm{ATEM}^{\mathrm{O}}(t, s, e)$} \label{table:once}
	\centering
	\begin{tabular}{rrrrrrrrrr}
		\toprule
		\multicolumn{1}{c}{$N$} & \multicolumn{1}{c}{$T$} & \multicolumn{1}{c}{$t$} & \multicolumn{1}{c}{$s$} & \multicolumn{1}{c}{$e$} & \multicolumn{1}{c}{Bias} & \multicolumn{1}{c}{RMSE} & \multicolumn{1}{c}{PW.CP} & \multicolumn{1}{c}{U.CP} & \multicolumn{1}{c}{CI.L}\\
		\midrule
		250 & 4 & 2 & 1 & 1 & -0.004 & 0.343 & 0.975 & 0.905 & 1.541\\
		250 & 4 & 3 & 1 & 1 & -0.008 & 0.383 & 0.965 & 0.905 & 1.652\\
		250 & 4 & 4 & 1 & 1 & 0.001 & 0.433 & 0.958 & 0.905 & 1.820\\
		\midrule
		1000 & 4 & 2 & 1 & 1 & -0.002 & 0.173 & 0.979 & 0.936 & 0.798\\
		1000 & 4 & 3 & 1 & 1 & -0.002 & 0.191 & 0.978 & 0.936 & 0.878\\
		1000 & 4 & 4 & 1 & 1 & -0.001 & 0.219 & 0.977 & 0.936 & 0.989\\
		\midrule
		4000 & 4 & 2 & 1 & 1 & -0.001 & 0.086 & 0.982 & 0.948 & 0.403\\
		4000 & 4 & 3 & 1 & 1 & -0.002 & 0.096 & 0.982 & 0.948 & 0.449\\
		4000 & 4 & 4 & 1 & 1 & -0.002 & 0.108 & 0.981 & 0.948 & 0.511\\
		\bottomrule
	\end{tabular}
\end{table}

\clearpage

\begin{table}[p]
	\caption{Monte Carlo simulation results: The event specification $\mathrm{ATEM}^{\mathrm{E}}(t, s, e)$} \label{table:event}
	\centering
	\begin{tabular}{rrrrrrrrrr}
		\toprule
		\multicolumn{1}{c}{$N$} & \multicolumn{1}{c}{$T$} & \multicolumn{1}{c}{$t$} & \multicolumn{1}{c}{$s$} & \multicolumn{1}{c}{$e$} & \multicolumn{1}{c}{Bias} & \multicolumn{1}{c}{RMSE} & \multicolumn{1}{c}{PW.CP} & \multicolumn{1}{c}{U.CP} & \multicolumn{1}{c}{CI.L}\\
		\midrule
		250 & 4 & 2 & 1 & 2 & -0.004 & 0.343 & 0.989 & 0.888 & 1.775\\
		250 & 4 & 3 & 1 & 2 & -0.011 & 0.445 & 0.981 & 0.888 & 2.143\\
		250 & 4 & 4 & 1 & 2 & 0.000 & 0.547 & 0.971 & 0.888 & 2.478\\
		250 & 4 & 3 & 2 & 3 & -0.002 & 0.413 & 0.990 & 0.888 & 2.172\\
		250 & 4 & 4 & 2 & 3 & 0.004 & 0.466 & 0.988 & 0.888 & 2.400\\
		250 & 4 & 4 & 3 & 4 & 0.000 & 0.493 & 0.988 & 0.888 & 2.566\\
		\midrule
		1000 & 4 & 2 & 1 & 2 & -0.002 & 0.173 & 0.993 & 0.931 & 0.924\\
		1000 & 4 & 3 & 1 & 2 & -0.003 & 0.224 & 0.991 & 0.931 & 1.177\\
		1000 & 4 & 4 & 1 & 2 & -0.001 & 0.289 & 0.989 & 0.931 & 1.438\\
		1000 & 4 & 3 & 2 & 3 & -0.001 & 0.204 & 0.992 & 0.931 & 1.101\\
		1000 & 4 & 4 & 2 & 3 & -0.002 & 0.231 & 0.991 & 0.931 & 1.239\\
		1000 & 4 & 4 & 3 & 4 & 0.001 & 0.239 & 0.992 & 0.931 & 1.297\\
		\midrule
		4000 & 4 & 2 & 1 & 2 & -0.001 & 0.086 & 0.993 & 0.941 & 0.467\\
		4000 & 4 & 3 & 1 & 2 & -0.002 & 0.114 & 0.994 & 0.941 & 0.611\\
		4000 & 4 & 4 & 1 & 2 & -0.002 & 0.144 & 0.993 & 0.941 & 0.775\\
		4000 & 4 & 3 & 2 & 3 & 0.001 & 0.103 & 0.992 & 0.941 & 0.554\\
		4000 & 4 & 4 & 2 & 3 & 0.000 & 0.115 & 0.993 & 0.941 & 0.627\\
		4000 & 4 & 4 & 3 & 4 & -0.004 & 0.119 & 0.992 & 0.941 & 0.651\\
		\bottomrule
	\end{tabular}
\end{table}

\clearpage

\begin{table}[p]
	\caption{Monte Carlo simulation results: The number specification $\mathrm{ATEM}^{\mathrm{N}}(t, s, e)$} \label{table:number}
	\centering
	\begin{tabular}{rrrrrrrrrr}
		\toprule
		\multicolumn{1}{c}{$N$} & \multicolumn{1}{c}{$T$} & \multicolumn{1}{c}{$t$} & \multicolumn{1}{c}{$s$} & \multicolumn{1}{c}{$e$} & \multicolumn{1}{c}{Bias} & \multicolumn{1}{c}{RMSE} & \multicolumn{1}{c}{PW.CP} & \multicolumn{1}{c}{U.CP} & \multicolumn{1}{c}{CI.L}\\
		\midrule
		250 & 4 & 2 & 1 & 1 & -0.004 & 0.343 & 0.987 & 0.874 & 1.688\\
		250 & 4 & 3 & 1 & 1 & -0.004 & 0.359 & 0.985 & 0.874 & 1.794\\
		250 & 4 & 4 & 1 & 1 & 0.002 & 0.393 & 0.986 & 0.874 & 1.960\\
		250 & 4 & 3 & 1 & 2 & -0.012 & 0.538 & 0.967 & 0.874 & 2.332\\
		250 & 4 & 4 & 1 & 2 & -0.002 & 0.506 & 0.977 & 0.874 & 2.329\\
		250 & 4 & 4 & 1 & 3 & 0.006 & 0.712 & 0.943 & 0.874 & 2.770\\
		\midrule
		1000 & 4 & 2 & 1 & 1 & -0.002 & 0.173 & 0.988 & 0.923 & 0.871\\
		1000 & 4 & 3 & 1 & 1 & -0.001 & 0.178 & 0.987 & 0.923 & 0.907\\
		1000 & 4 & 4 & 1 & 1 & 0.001 & 0.193 & 0.991 & 0.923 & 0.987\\
		1000 & 4 & 3 & 1 & 2 & -0.003 & 0.284 & 0.985 & 0.923 & 1.337\\
		1000 & 4 & 4 & 1 & 2 & -0.003 & 0.252 & 0.985 & 0.923 & 1.242\\
		1000 & 4 & 4 & 1 & 3 & -0.002 & 0.409 & 0.974 & 0.923 & 1.739\\
		\midrule
		4000 & 4 & 2 & 1 & 1 & -0.001 & 0.086 & 0.989 & 0.942 & 0.439\\
		4000 & 4 & 3 & 1 & 1 & -0.001 & 0.089 & 0.991 & 0.942 & 0.455\\
		4000 & 4 & 4 & 1 & 1 & -0.002 & 0.096 & 0.990 & 0.942 & 0.495\\
		4000 & 4 & 3 & 1 & 2 & -0.003 & 0.147 & 0.988 & 0.942 & 0.723\\
		4000 & 4 & 4 & 1 & 2 & -0.002 & 0.125 & 0.990 & 0.942 & 0.640\\
		4000 & 4 & 4 & 1 & 3 & -0.002 & 0.219 & 0.987 & 0.942 & 1.036\\
		\bottomrule
	\end{tabular}
\end{table}

\end{document}